\def\icalp{0} 

\ifnum\icalp=1
\documentclass[a4paper,UKenglish,cleveref, autoref, thm-restate,authorcolumns]{lipics-v2019}
\else
\documentclass[11pt]{article}
\usepackage[margin=1in]{geometry}
\usepackage[colorlinks,urlcolor=blue,citecolor=blue,linkcolor=blue]{hyperref}
\usepackage{fullpage}
\usepackage{subcaption}
\fi

\usepackage{graphicx}
\usepackage{epstopdf}
\graphicspath{{./figs/}}

\usepackage{multirow}
\usepackage{comment}
\usepackage{amsfonts,amsmath,mathtools,amssymb}
\usepackage{xcolor}
\usepackage[linesnumbered,ruled,vlined]{algorithm2e}

\SetCommentSty{mycommfont}

\SetKwInput{KwInput}{Input}                
\SetKwInput{KwOutput}{Output}              

\usepackage{bbm,bm}
\usepackage{xspace,prettyref}
\usepackage{color}
\usepackage{dsfont}
\ifnum\icalp=0
\newtheorem{theorem}{Theorem}
\newtheorem{lemma}{Lemma}
\newtheorem{claim}[lemma]{Claim}
\newtheorem{definition}[lemma]{Definition}
\newtheorem{corollary}[lemma]{Corollary}

\fi

\newtheorem{observation}{Observation}

\newcommand{\pr}{{\rm Pr}}

\newcommand{\dist}{\textrm{dist}}

\def\GI{{\sc GI}}

\def\congest{{\sf CONGEST}}
\def\local{{\sf LOCAL}}
\def\yes{{\sf yes}}
\def\no{{\sf no}}
\def\accept{{\sf accept}}
\def\reject{{\sf reject}}

\def\Ex{{\rm Ex}}
\def\NN{{\mathbb{N}}}

\ifnum\icalp=0
\newcommand{\qed}{\;\;\;\FullBox}
\newenvironment{proof}{\noindent{\bf Proof:~~}}{\(\qed\)}
\fi

\newcommand{\eqdef}{~\triangleq~}
\def\eps{\epsilon}

\def\poly{{\rm poly}}

\def\diam#1{D( {#1} )}

\def\FullBox{\hbox{\vrule width 8pt height 8pt depth 0pt}}

        {$\blacksquare$\par}

\newcommand{\mnote}[1]{}

\definecolor{darkgreen}{rgb}{0.0, 0.5, 0.0}

\newcommand{\Reut}[1]{}
\newcommand{\moti}[1]{}
\newcommand{\Moti}[1]{}

\newcommand{\mchange}[1]{}
\newcommand{\rchange}[1]{}

\ifnum\icalp=1
\title{Distributed Testing of Graph Isomorphism in the \congest\ model}

\author{Reut Levi}{Efi Arazi School of Computer Science,  The Interdisciplinary Center, Israel}{reut.levi1@idc.ac.il}{}{}
\author{Moti Medina}{School of Electrical \& Computer Engineering, Ben-Gurion University of the Negev, Israel}{medinamo@bgu.ac.il}{https://orcid.org/0000-0002-5572-3754}{This research was supported by the Israel Science Foundation under Grant 867/19.}

\authorrunning{R. Levi and M. Medina} 

\titlerunning{Dist. Property Testing of Graph Iso. in the \congest\ model} 

\Copyright{Reut Levi and Moti Medina} 



\EventEditors{John Q. Open and Joan R. Access}
\EventNoEds{2}
\EventLongTitle{42nd Conference on Very Important Topics (CVIT 2016)}
\EventShortTitle{CVIT 2016}
\EventAcronym{CVIT}
\EventYear{2016}
\EventDate{December 24--27, 2016}
\EventLocation{Little Whinging, United Kingdom}
\EventLogo{}
\SeriesVolume{42}
\ArticleNo{23}
\fi

\begin{document}
\setcounter{page}{0}
\ifnum\icalp=0
\title{Distributed Testing of Graph Isomorphism in the \congest\ model}

\author{
Reut Levi\thanks{Efi Arazi School of Computer Science,  The Interdisciplinary Center, Israel.
  Email:~{\tt reut.levi1@idc.ac.il}.}
\and
Moti Medina\thanks{School of Electrical \& Computer Engineering, Ben-Gurion University of the Negev, Israel.
  Email:~{\tt medinamo@bgu.ac.il}. This research was supported by the Israel Science Foundation under Grant 867/19.}
}

\date{}
\fi
\maketitle
\begin{abstract}
In this paper we study the problem of testing graph isomorphism (\GI) in the \congest\ distributed model.
In this setting we test whether the distributive network, $G_U$, is isomorphic to $G_K$ which is given as an input to all the nodes in the network, or alternatively, only to a single node.

We first consider the \emph{decision} variant of the problem in which the algorithm should distinguish the case where $G_U$ and $G_K$ are isomorphic from the case where $G_U$ and $G_K$ are not isomorphic.
Specifically, if $G_U$ and $G_K$ are not isomorphic then w.h.p. at least one node should output \reject\ and otherwise all nodes should output \accept .
We provide a randomized algorithm with $O(n)$ rounds for the setting in which $G_K$ is given only to a single node.
We prove that for this setting the number of rounds of any deterministic algorithm is $\tilde{\Omega}(n^2)$ rounds, where $n$ denotes the number of nodes, which implies a separation between the randomized and the deterministic complexities of deciding \GI .
Our algorithm can be adapted to the semi-streaming model, where a single pass is performed and  $\tilde{O}(n)$ bits of space are used.

We then consider the \emph{property testing} variant of the problem, where the algorithm is only required to distinguish the case that $G_U$ and $G_K$ are isomorphic from the case that $G_U$ and $G_K$ are  \emph{far} from being isomorphic (according to some predetermined distance measure).
We show that every (possibly randomized) algorithm, requires $\Omega(D)$ rounds, where $D$ denotes the diameter of the network. This lower bound holds even if all the nodes are given $G_K$ as an input, and even if the message size is unbounded. We provide a randomized algorithm with an almost matching round complexity of $O(D+(\eps^{-1}\log n)^2)$ rounds that is suitable for dense graphs (namely, graphs with $\Omega(n^2)$ edges).

We also show that with the same number of rounds it is possible that each node outputs its mapping according to a bijection which is an {\em approximated} isomorphism.

We conclude with simple simulation arguments that allow us to adapt centralized property testing algorithms and obtain essentially tight algorithms with round complexity  $\tilde{O}(D)$ for special families of sparse graphs.

\end{abstract}

\clearpage

\section{Introduction}\label{sec:intro}
Testing graph isomorphism is one of the most fundamental computational problems in graph theory.
A pair of graphs $G$ and $H$ are isomorphic if there is a bijection that maps the nodes of $G$ to the nodes of $H$
such that every edge of $G$ is mapped to an edge of $H$ and likewise for non-edges.
Currently, it is not known whether there exists an efficient algorithm for this problem and in fact it is one of the few natural problems which is a candidate for being in NP-intermediate, that is, neither in P nor NP-complete.
In order to obtain efficient algorithms for this problem, various restrictions and relaxations were considered (e.g.~\cite{kelly1957congruence,hopcroft1974linear}).
This problem has been extensively studied also in other computational models such as parallel computation models~\cite{grohe2006testing,luks1986parallel,chen1996graph,jaja1988parallel,verbitsky2007planar,kobler2006graph, gazit1990randomized,chen1994parallel} and in the realm of property testing in which the main complexity measure is the query complexity~\cite{fischer2008testing,DBLP:conf/stoc/OnakS18,DBLP:journals/eccc/Goldreich19b,NS13,KY14,BKN16}.

In the context of distributive models such as the \congest ~\cite{P00} and the \local ~\cite{L92} models, the main complexity measure is the round complexity and the computational power is usually considered to be unbounded.
Therefore in these models the complexity of the problem may change dramatically.
While there seem to be many sensible settings, one of the simplest settings of the problem for distributive models is to test for isomorphism between the distributed network, $G_U$, and a known graph, $G_K$, which is given as an input to all the nodes in the network, or alternatively, only to a subset of the nodes~\footnote{This formulation in which $G_K$ is a parameter falls into the category of massively parameterized problems and is also considered in the setting of property testing~\cite{fischer2008testing,DBLP:journals/eccc/Goldreich19b}.}.
The requirement from the algorithm is that if $G_K$ and $G_U$ are isomorphic, then with high probability~\footnote{We say that an algorithm succeed with high probability, if it succeeds with probability at least $1- 1/n^c$ for any constant $c$ (without changing the round complexity asymptotically.} all nodes should output \accept\ and that at least one node should output \reject\ otherwise.

Since the property of being isomorphic to a specific graph is inherently global, intuitively we expect the round complexity to be $\Omega(D)$ where $D$ denotes the diameter of the network (even for the case in which $G_K$ is given as an input to all the nodes in the network).
As we show, this intuition is correct even for the \local\  model, in which there is no bound on the message size.
Therefore, in the \local\ model, it is not possible to improve over the trivial algorithm that collects the entire information on the network at a single node in $O(D)$ rounds and tests for graph isomorphism in a centralized manner.
In the \congest\ model, in which the message-size is bounded by $O(\log n)$, where $n$ denotes the number of nodes in the network, implementing this trivial solution may require $O(n^2)$ rounds.
This leads to the obvious question whether is it possible to obtain round complexity which is better than $O(n^2)$ in the \congest\ model.

Another interesting question is whether we can obtain better bounds if we relax the decision problem (as considered in the realm of {\em property testing\/})
such that the algorithm is only required to distinguish between pairs of graphs which are isomorphic and pairs of graphs which are {\em far} from being isomorphic (according to some predetermined distance measure).

In this setting we define the problem as follows.
Let $G_U$ be the distributed network and let $m$ denote the number of edges in the network or an upper bound on this number.
We say that a pair of graphs are $\epsilon$-far from being isomorphic if $\epsilon m$ edges need to be deleted/inserted in order to make the graphs isomorphic, where $m$ denotes the number of edges in the network.
An adversarially chosen node, $r$, receives as an input the graph $G_K$ and a proximity parameter $\epsilon\in (0,1)$.
The requirement from the algorithm is as follows.
If $G_K$ and $G_U$ are isomorphic, then w.h.p. all nodes should output \accept.
If $G_K$ and $G_U$ are $\epsilon$-far from being isomorphic, then w.h.p. at least one node should output \no.

\subsection{Our Results}\label{res.sec}
In this section we outline our results. We further elaborate on our results in the next sections.  In all that follows, unless explicitly stated otherwise, when we refer to distributed algorithms, we mean in the \congest\ model.

\subparagraph*{A Decision Algorithm and a Lower Bound for the Decision Problem. }
For the (exact) decision problem we provide a randomized one-sided error algorithm that runs in $O(n)$ rounds and succeeds
\ifnum\icalp=1
with high probability (see Theorem~\ref{thm:exact} in Appendix~\ref{sec:exact}).
\else
with high probability (see Theorem~\ref{thm:exact}).
\fi
The algorithm works even for the setting in which $G_K$ is given only to a single node (that may be chosen adversely).
For this setting we prove that any deterministic algorithm requires $\tilde{\Omega}(n^2)$ rounds, which implies a separation between the randomized and the deterministic complexity of the
\ifnum\icalp=1
decision problem (see Theorem~\ref{thm:detlb} in Appendix~\ref{sec:lb}).
\else
decision problem (see Theorem~\ref{thm:detlb}).
\fi
We note that our algorithm can be adapted to the semi-streaming model~\cite{feigenbauma2005graph} in which it uses $\tilde{O}(n)$ bits of space and performs only one-pass (see Theorem~\ref{thm:sstream}).

\subparagraph*{A Lower Bound for the Property Testing Variant. }
For property testing algorithms we
show that even under this relaxation $\Omega(D)$ rounds are necessary even for constant $\eps$ and constant error probability.
This lower bound holds even in the \local\ model, for two-sided error algorithms, and even if all the nodes receive $G_K$ as
\ifnum\icalp=1
an input (see Theorem~\ref{thm:lb} in Appendix~\ref{sec:lb}).
\else
an input (see Theorem~\ref{thm:lb}).
\fi
It also holds for dense graphs, namely, when $m = \Theta(n^2)$ and for sparse graphs, that is when $m = \Theta(n)$.

\subparagraph*{A Property Testing Algorithm and Computation of Approximated Isomorphism. }
We provide a distributed two-sided error property testing algorithm that runs in $O(D+(\eps^{-1}\cdot\log n)^2)$ rounds and succeeds with high probability for the case that $m = \Theta(n^2)$, implying that our result is tight up to an additive term of $(\eps^{-1}\cdot\log n)^2$ (see Theorem~\ref{thm:test-iso}).
This algorithm works even in the setting in which $G_K$ is given only to a single node.
We note that the graphs that are constructed for the lower bound for the exact variant are dense and have a constant diameter.
Therefore for these graphs, the property testing algorithm runs in only $O((\eps^{-1}\cdot\log n)^2)$ rounds (while the decision algorithm runs in $O(n)$ rounds).

If $G_K$ is given to all the nodes and the graphs are indeed isomorphic then we show that we can also approximately recover the isomorphism with the same round complexity as of testing.
Specifically, each node $v$ outputs $g(v)$ where $g$ is a bijection such that the graph $g(G_U)$, namely the graph in which we re-name the nodes according to $g$, is $\epsilon$-close to $G_K$.

\subparagraph*{Simulation Arguments and their application to special families of sparse graphs. }
Finally, we show, by simple simulation arguments,  that it is possible to obtain essentially tight algorithms with $\tilde{O}(D)$ round complexity for special families of sparse graphs by adapting centralized property testing algorithms.
In particular, these algorithms apply for bounded-degree minor-free graphs and general outerplanar graphs.
\ifnum\icalp=1
(see Appendix~\ref{sec:simulc}).
\fi

\subsection{The Decision Algorithm}
As described above, a naive approach for testing isomorphism to $G_K$ is to gather the entire information on the network at a single node and then to test for isomorphism in a centralized manner.
By the brute-force approach, we may go over all possible bijections between the nodes of the graphs and test for equality between the corresponding graphs.
Our algorithm follows this approach with the difference that it only gathers a compressed version of the network as in the algorithm of Abboud et al.~\cite{DBLP:journals/corr/abs-1901-01630} for the Identical Subgraph Detection problem.
The idea of their algorithm is to reduce the problem of testing if two graphs are equal to the problem of testing equality between a pair of binary strings.
From the fact that the test for equality has a one-sided error, namely it never rejects identical graphs, it follows that our algorithm never rejects isomorphic graphs.
To ensure that our algorithm is sound we amplify the success probability of the equality test and, as a result, obtain a total round complexity of $O(n)$.

\subsection{A Lower Bound for the Decision Problem}
We reduce Set-Equality to the problem of deciding isomorphism in the setting in which only a single node receives $G_K$ as an input (as it is the case for our upper bound).
The idea is to construct a graph $G_{x, y}$ over $n$ nodes for every pair of strings $x, y \in \{0,1\}^k$  where $k = \Theta(n^2)$ such that $G_{x, y}$ is isomorphic to $G_{x', y'}$ if and only if $x = x'$ and $y = y'$.
Let $x$ and $y$ denote the input of Alice and Bob, respectively. In the reduction, $G_K$ is known to Alice and is taken to be $G_{x,x}$. Alice and Bob simulate the distributed algorithm on the graph $G_{x, y}$, which by construction is isomorphic to $G_{x, x}$ if and only if $x = y$, as desired.
This reduction yields a lower bound of $\Omega(n^2/\log n)$ rounds for any deterministic algorithm.

\subsection{A High-Level Description of the Property Testing Algorithm}

Our algorithm closely follows the approach taken by Fischer and Matsliah~\cite{fischer2008testing} for testing graph isomorphism in the dense-graph model~\cite{goldreich1998property} with two sided-error.
However, in order to obtain a round complexity which only depends poly-logarithmically in $n$ (rather than a dependency of $\tilde{O}(\sqrt{n})$ as the query complexity in~\cite{goldreich1998property}), we need to diverge from their approach as described next.

\subsubsection{The Algorithm of Fischer-Matsliah}
\sloppy
The algorithm of Fischer-Matsliah begins with picking, u.a.r., a sequence of $s = \poly(\eps^{-1}, \log(n))$ nodes from the unknown graph.
The selection of these nodes induces labels for each node in the graph as follows.
The label of each node $v$ is a string of $s$ bits where the $i$-th bit indicates whether $v$ is a neighbor of the $i$-th node in the sequence.
This labeling scheme guarantees that, with high probability, only ``similar'' nodes, that is, nodes with similar sets of neighbors, might have identical labels.
It is not hard to see that if the graphs are isomorphic, then given that we managed to map the nodes in the sequence according to the isomorphism, both graphs should have the same frequency over labels.
More surprisingly, it is shown by Fischer and Matsliah that if the nodes in the sequence are mapped according to the isomorphism then it is possible to extend this mapping on-the-fly and obtain, roughly speaking, an approximate isomorphism.
In particular, they showed that as long as each node in the graph is mapped to a node with the same label in the other graph (with respect to the mapped sequence), then the obtained function is close to being an isomorphism.
This is due to the fact that nodes which are too ``different'' are likely to have different labels and the fact that similar nodes are exchangeable.
Given a candidate for the approximate isomorphism, the problem is then reduced to testing closeness of graphs.
Therefore, if the graphs are isomorphic then by going over all possible mappings of the selected sequence (there are only quasi-polynomial many ways to map these nodes) and extending this partial mapping as described above, one should be able to obtain a function, $f$, which is close to being an isomorphism.
On the other hand, if the graphs are far from being isomorphic then by definition any bijection gives two graphs which are far from each other.
Therefore, these two cases can be distinguished by approximating the Hamming distance of the corresponding adjacency-matrices.
In turn, this can be done by selecting random locations (that is, potential edges) and checking the values of both matrices in these locations.
Since constructing $f$ entirely would be too costly, one needs to be able to generate $f$ on-the-fly.
A crucial point is that its generation can not depend on the selection of the random locations.
In other words, its generation should be query-order-oblivious.
To this end, in the algorithm of Fischer-Matsliah they first test if the distributions over the labels are close.
If so, they can safely generate $f$ on-the-fly while ensuring that there is only little dependency between $f$ and the queries the algorithm makes to $f$.
This is done by simply mapping a node $v$ to a random node in $G_K$ that have the same label as $v$.~\footnote{The little dependency between $f$ and the queries that the algorithm makes to $f$ comes from the fact that the frequencies of labels of the two graphs are not necessarily identical as they are only guaranteed to be close (w.h.p.).}
The query complexity of testing closeness of distributions, which is $\tilde{O}(\sqrt{n})$, dominates the query complexity of the algorithm.
As shown in~\cite{fischer2008testing}, in the centralized setting this algorithm is essentially tight.

\subsubsection{Our Algorithm}
In the \congest\ model, by straight-forward simulation arguments it follows that one can simulate the algorithm of Fischer-Matsliah in $\tilde{O}(D + \sqrt{n})$ rounds by collecting the answers to the queries of the algorithm at a single node and simulating the centralized algorithm (see Claim~\ref{clm:sim}).
A crucial observation for improving this bound is that nodes that have the same label also have at least one neighbor in common (with the only exception of the all-zero label), therefore they can be coordinated by one of their common neighbors.
Moreover, it is possible to obtain both samples and access to the frequencies of the labels via these coordinators.
Our algorithm proceeds as follows. As in~\cite{fischer2008testing} a sequence $C$ of $s$ random nodes is selected and is sent to the entire network (in $O(D)$ rounds).
Each node figures out its label and broadcasts this label to its neighbors.
The node $r$, that received $G_K$ as an input, selects a random set of potential edges $(i_1, j_1), \ldots, (i_k, j_k)$, where $k = \poly(\log n, \eps^{-1})$. It then broadcasts this set to the entire network.
For each potential edge, the information whether it is an actual edge in the network is sent to $r$.
For each label of a node in $I =
\{i_1, j_1, \ldots, i_k, j_k\}$, the corresponding coordinator sends to $r$ the frequency of this label.
From this point the rest of the computation is done centrally at $r$.
We say that a sequence, $P$, of nodes in $G_K$ is {\em good} with respect to $C$ and $I$ if it induces the same frequency of labels as $C$ when restricted to labels of nodes in $I$.
The node $r$ goes over all possible mappings of $C$ to the known graph and looks for good sequences (with respect to $C$ and $I$).
For every good sequence, $P$, $r$ generates a function $f$ on-the-fly: on query $v$, $f$ maps $v$ to a random node in $V_K$ which is still unmatched and has the same label as $v$ (w.r.t. $P$).
As we show, from the fact that the sequence is good it follows that $f$ is query-order-oblivious.
Let $f(G_U)$ denotes the graph obtained from $G_U$ after applying $f$ on $V_K$.
As in the algorithm of Fischer-Matsliah, if the graphs are isomorphic and the sequence $P$ is the mapping of $C$ according to the isomorphism, then $f(G_U)$ is guaranteed to be (w.h.p.) $\epsilon$-close $G_K$.
The set of potential edges is then used to approximate the distance between $G_K$ and $f(G_U)$.
This allows us to obtain a significant improvement in the round complexity (over the straight-forward simulation), in terms of $n$, from $\tilde{O}(\sqrt{n})$ to $O(\log^2 n)$.
\label{subsec-intro-high-level}

\subsection{High-level Approach for Computing an Approximated Isomorphism}
As described in the previous section, if the graphs are isomorphic then w.h.p. the algorithm finds a sequence $P$ that corresponds to a bijection $f$ such that $f(G_U)$ is guaranteed to be (w.h.p.) $\epsilon$-close to $G_K$.
The algorithm accesses $f$ only on a small set of random locations.
It is tempting to try to output $f(v)$ for every $v$ in the network.
Assume now that every node in the network knows $G_K$.
If the sequence $P$ is indeed the mapping of $C$ according to an isomorphism then the following naive approach should work.
Each coordinator can independently map to $V_K$ the nodes that are assigned to it according to their labels.
However, it might be the case that $P$ is not the mapping of $C$ according to any isomorphism (although it passed the test).
In particular it might be that it is not good with respect to $C$ and $V_U$ (recall that $P$ is good w.r.t. $C$ and $I$).
In this case we may want the coordinators of the nodes to be coordinated such that they exchange the mapping of nodes with ``underflow'' and ``overflow'' labels.
Since there might be $O(n)$ labels, such coordination might cause too much congestion.
To this end we cluster the labels according to their most significant bit and assign a single coordinator to each cluster.
Since there are only $\poly(\eps^{-1}, \log(n))$ many clusters, these coordinators can coordinate without causing too much congestion.
The main technicality that needs to be addressed is showing that the resulting mapping, $g$, is close enough to $f$ and hence is an approximated isomorphism.
We prove this by coupling $g$ and $f$ and showing that they agree on the mapping of most nodes.

\subsection{A Lower Bound for the Property Testing Variant}
We prove that for any $D$ there exists a family of graphs with diameter $\Theta(D)$ such that any distributed two-sided error property testing algorithm for testing isomorphism on this family of graphs
requires $\Omega(D)$ rounds.
In  the construction we start with a pair of graphs $G_1$ and $G_2$ that have diameter $O(D)$ which are far from being isomorphic.
The graph $G_U$ is then defined to be composed of $G_1$ and $G_2$ and a path of length $\Theta(D)$ that connects the two graphs.
Roughly speaking, the idea is to argue that for round complexity which is at most $D/c$, where $c$ is some absolute constant, the nodes in $G_U$ which belong to the side of $G_1$ cannot distinguish the case in which the network is composed of two graphs which are isomorphic to $G_1$ (connected by a path). Likewise for the nodes that belong to the side of $G_2$ (that cannot distinguish the case in which the network is composed of two graphs which are isomorphic to $G_2$).
It then follows that the algorithm must err.
In the detailed proof, which appears in the appendix, there are some technicalities that need to be addressed in order to prove that the above argument still holds when the nodes may use randomness, port numbers and IDs.
\def\local{{\sf LOCAL}}

\subsection{Related work}
In this section we overview results in distributed decision and property testing in the \congest\ model.
We also overview related results in centralized property testing.

\subparagraph*{Distributed Decision.}
There is a large body of algorithms and lower bounds for the  \emph{subgraph detection problem}: given a fixed graph $H$, and an input graph $G$, the task is to decide whether $G$ contains a subgraph which is isomorphic to $H$.
The subgraphs considered include: paths~\cite{korhonen2018deterministic}, cycles~\cite{korhonen2018deterministic,fischer2018possibilities,eden2019sublinear}, triangles~\cite{izumi2017triangle,abboud2017fooling,chang2019distributed}, cliques~\cite{czumaj2019detecting,eden2019sublinear,DBLP:conf/icalp/BonneC19}.
Abboud et al.~\cite[Sec.~6.2]{DBLP:journals/corr/abs-1901-01630} considered the
\emph{identical subgraph detection problem}. In this problem the graph's nodes are partitioned into two equal sets. The task is to decide whether the induced graphs on these two sets are identical w.r.t. to a fixed mapping between the nodes of these two sets. They showed an $\Omega(n^2)$ lower bound on the number of rounds of any deterministic algorithm and a randomized algorithm that performs $O(D)$ rounds which succeeds w.h.p.

\subparagraph*{Distributed Property Testing for Graph Problems.}
Distributed property testing was initiated by Censor-Hillel et al.~\cite{CFSV16}.
In particular, they designed and analyzed distributed property testing algorithms for: triangle-freeness, cycle-freeness, and bipartiteness. They also proved a logarithmic lower bound for the latter two properties.
While they mainly focus on the bounded degree model and the general model they also studied the dense model. In this model they showed that for a certain class of problems, any centralized property testing algorithm can be emulated in the distributed model such that number of rounds is $q^2$ where $q$ denotes the number of queries made by the centralized tester.
Fraigniaud et al.~\cite{FRST16} studied distributed property testing of excluded subgraphs of size $4$ and $5$.
Since the appearance of the above papers, there was a
fruitful line of research in distributed property testing for various properties, mainly focusing on properties of whether a graph excludes a fixed sub-graph
~\cite{FO17,FGO17c,FMORT17,ELM17,DISC17,FGO17a}.
Other problems on graphs such as testing planarity, and testing the conductance was studied in~\cite{DBLP:conf/podc/LeviMR18,FY17}, respectively.


\subparagraph*{Centralized Property Testing. }
Fischer and Matsliah~\cite{fischer2008testing} studied the graph isomorphism problem in the dense-graph model~\cite{goldreich1998property}.\footnote{In the dense-graph model, a graph $G$ is considered to be $\eps$-far from a property $\Pi$ if the symmetric difference between its edge set to the edge set of any graph in $\Pi$ is greater than $\eps |V(G)|^2$.}
They considered four variations of the Graph Isomorphism testing problem: (1)~one-sided error, where one of the graphs is known, and there is a query access to the graph which is tested, i.e., the tested graph is ``unknown'', (2)~one-sided error, where there is a query access for both graphs, i.e., both graphs are unknown, (3)~two-sided error, where one graph is known, (4)~ two sided error, where both graphs are unknown. For the first three variants Fischer and Matsliah~\cite{fischer2008testing} showed (almost) matching lower and upper bounds of, respectively: (1)~$\tilde{O}(n)$, $\Omega(n)$, (2)~$\tilde{O}(n^{3/2})$, $\Omega(n^{3/2})$, and  (3)~$\tilde{O}(n^{1/2})$, $\Omega(n^{1/2})$, where $n$ is the number of vertices of each input (known or unknown) graph. For the fourth variant they showed an upper-bound of $\tilde{O}(n^{5/4})$ and a lower-bound of $\Omega(n)$. Onak and Sun~\cite{DBLP:conf/stoc/OnakS18} improved the upper bound of the fourth case to $O(n)\cdot 2^{\tilde{O}\left(\sqrt{\log n}\right)}$ by bypassing the distribution testing reduction that was used by~\cite{fischer2008testing}.
Property testing of graph isomorphism was also considered in the bounded-degree model~\cite{GR02}\footnote{In the bounded-graph model, a graph $G$ with maximum degree $d$, is considered to be $\eps$-far from a property $\Pi$ if the symmetric difference between its  edge set to the edge set of any graph in $\Pi$ is greater than $\eps d|V(G)|$ }.
Goldreich~\cite{DBLP:journals/eccc/Goldreich19b} proved that the query complexity of any property testing algorithm is at least $\tilde{\Omega}(n^{1/2})$, for the variant in which one graph is known, and $\tilde{\Omega}(n^{2/3})$ when both graphs are unknown.
Newman and Sohler~\cite{NS13} provide an algorithm for minor-free graphs with degree bounded by $d=O(1)$ (this class includes for example bounded degree planar graphs) whose query complexity is independent of the size of the graph. Moreover, they showed that any property is testable in this class of graphs with the same query complexity.
Kusumoto and Yoshida~\cite{KY14}, and Babu, Khoury, and Newman~\cite{BKN16} considered testing of isomorphism between graphs which are forests and outerplanar in the general model~\cite{KKR04}\footnote{In the general-graph model, a graph $G$, is considered to be $\eps$-far from a property $\Pi$ if the symmetric difference between its  edge set to the edge set of any graph in $\Pi$ is greater than $\eps |E(G)|$.}, respectively. They both proved an upper bound of $\poly\log n$ and a lower bound of $\Omega(\sqrt{\log n})$ was shown in~\cite{KY14}. Moreover, they proved that any graphs property is testable on these family of graphs with $\poly\log n$ queries.

\section{The Algorithm for Testing Isomorphism in Dense Graphs}\label{sec:iso-test}
In this section, we describe and analyze the distributed algorithm for testing graph isomorphism in dense graphs.
We begin with several useful definitions and observations, followed by the listing Algorithm~\ref{alg:testing} and the proof of its correctness (which follows from Lemma~\ref{lemma:accepts}
and Lemma~\ref{lemma:reject}). Finally we discuss in more details how the algorithm is implemented in the \congest\ model.

We establish the following theorem.
\begin{theorem}\label{thm:test-iso}
There exists a distributed two-sided error property testing algorithm for testing isomorphism (of dense graphs) that runs in
$O(D+(\eps^{-1}\log n)^2)$ rounds in the \congest\ model. The algorithm succeeds with high probability.
\end{theorem}

\subsection{Definitions and Notation}
We shall use the following definitions in our algorithm and in its analysis.

\medskip\noindent
Let $G$ be a graph and let $C = (c_1, \ldots, c_s)$ be a sequence of $s$ nodes from $V(G)$.
\begin{definition}[\cite{fischer2008testing}]\label{def1}
For every node $v\in V(G)$, the $C$-label of $v$ in $G$, denoted by $\ell^G_C(v)$, is a string of $s$ bits defined as follows: $$\ell^G_C(v)_i=1 \Leftrightarrow c_i \in N_G(v)\:,$$
where $N_G(v)$ denotes the neighbors of $v$ in $G$.
\end{definition}
%
\medskip\noindent
We use the $\triangle$-operator to denote both the symmetric difference between two sets and when applied on graphs it denotes the Hamming distance between the corresponding adjacency matrices.
\begin{definition}[\cite{fischer2008testing}]\label{def-sep}
For $\beta \in (0,1]$, we say that $C$ is $\beta$-separating if for
every pair of nodes $u, v$ such that $|\triangle(N(u), N(v))| \geq \beta n$ it holds that $u$ and $v$
have different $C$-labels in $G$.
\end{definition}
\begin{definition}[inverse of $\ell^G_C$]
For a label $x \in \{0,1\}^s$, define $S^G_C(x) \triangleq \{v\in V(G) : \ell^{G}_{C}(v) = x\}$. Namely, $S^G_C(x)$ is the set of nodes in $G$ for which the $C$-label is $x$.
\end{definition}
%
%
Let $G$ and $H$ be a pair of graphs such that $|V(G)| = |V(H)|$.
The following definitions are defined with respect to a pair of sequences of $s\geq 1$ nodes from $V(G)$ and $V(H)$, $C_G = (c^G_1, \ldots, c^G_s)$ and $C_H = (c^H_1, \ldots, c^H_s)$, respectively.

\medskip\noindent
We next define what we mean by saying that the mapping of a function $f:V(G)\rightarrow V(H)$ is consistent w.r.t. the labels of $C_G$ and $C_H$.

\begin{definition}\label{def2}
For $f:V(G)\rightarrow V(H)$ which is a bijection, we say that $f$ is {\em $(C_G, C_H)$-label-consistent} if
the following holds:
\begin{enumerate}
\item $f$ maps $C_G$ to $C_H$: $f(c^G_i) = c^H_i$ for every $i\in [s]$.
\item The label of a node and its image is the same: $\ell^{G}_{C_G}(v) = \ell^{H}_{C_H}(f(v))$ for every $v\in V(G)$.
\end{enumerate}
\end{definition}

\medskip\noindent
For $f:V(G)\rightarrow V(H)$ and a sequence $C = (c_1, \ldots, c_s)$, we define $f(C)$ to denote $(f(c_1), \ldots, f(c_s))$ and $f(G)$ to denote the graph whose nodes are $V(H)$ and its edge set is $\{\{f(u), f(v)\}: \{u, v\} \in E(G)\}$.

We next observe that if $G$ and $H$ are isomorphic than for any sequence $C$ and any function $f$ which is an isomorphism between $G$ and $H$, $f$ is consistent w.r.t. $C$ and $f(C)$.
\begin{observation}~\label{lemma:same}
If $G$ and $H$ are isomorphic and $\pi$ is an isomorphism from $V(G)$ to $V(H)$ then for every sequence of nodes , $C$, from $V(G)$, $\pi$ is $(C, \pi(C))$-label consistent.
In particular, $|S^G_{C}(x)| = |S^H_{\pi(C)}(x)|$ for every $x \in \{0,1\}^s$.
\end{observation}

If $f$ is not an isomorphism then it might be the case that it is not consistent w.r.t. $C$ and $f(C)$. We next define a weaker notion of consistency which is being maximally-label-consistent.

\begin{definition}\label{def3}
We say that a function $f:V(G)\rightarrow V(H)$ is {\em maximally $(C_G, C_H)$-label-consistent} if
the following holds:
\begin{enumerate}
\item $f$ is a bijection.
\item $f$ maps $C_G$ to $C_H$: $f(c^G_i) = c^H_i$ for every $i\in [s]$.
\item For every $x \in \{0,1\}^s$ such that $|S^G_{C_G}(x)| = |S^H_{C_H}(x)|$, $f$ maps the elements of $S^G_{C_G}(x)$ to the elements of $S^H_{C_H}(x)$.
\end{enumerate}
\end{definition}

\ifnum\icalp=1
See Figure~\ref{fig:labels} in Appendix~\ref{app:figs} for illustration for the definitions in this section.
\else
See Figure~\ref{fig:labels} for illustration for the definitions in this section.
\ifnum\icalp=1
\section{Figures}\label{app:figs}
\fi
\begin{figure}[h!]
    \centering
    \begin{subfigure}[b]{0.4\textwidth}
        \centering
        \includegraphics[width=\textwidth]{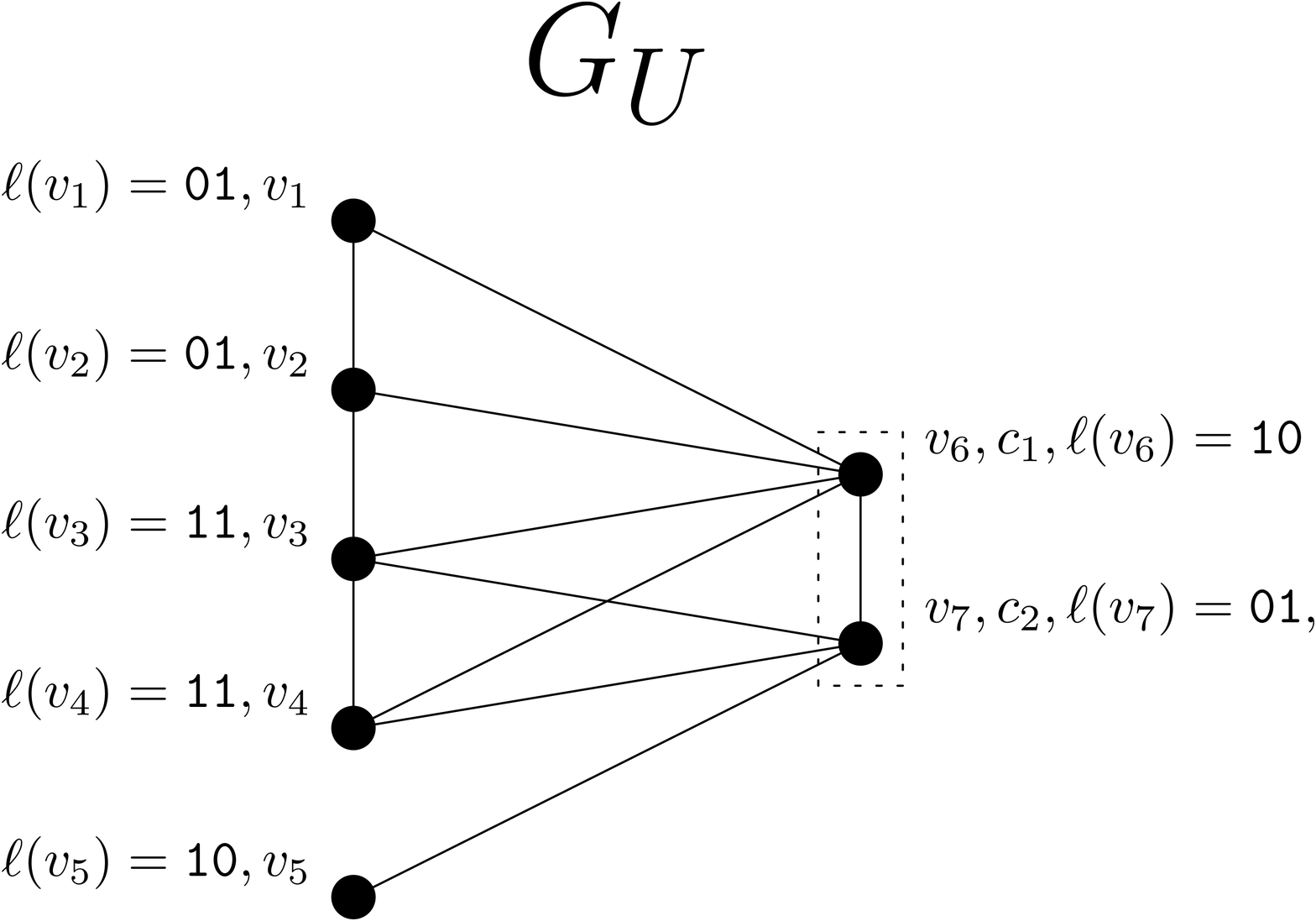}
        \caption{}
        \label{fig:Gu}
    \end{subfigure}
    \quad
    \begin{subfigure}[b]{0.4\textwidth}
        \centering
        \includegraphics[width=\textwidth]{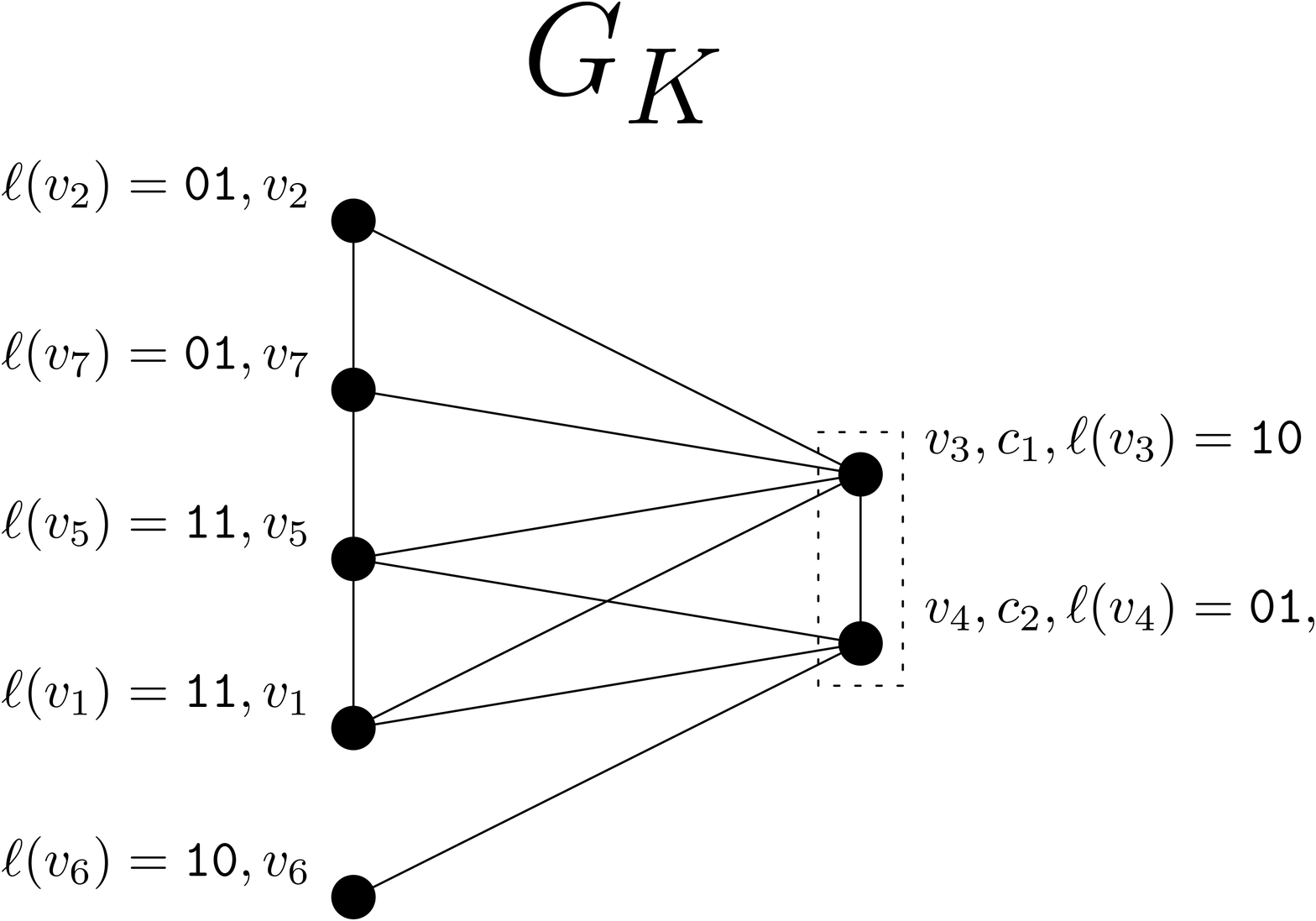}
        \caption{}
        \label{fig:Gk}
    \end{subfigure}
    \\
    \begin{subfigure}[b]{0.3\textwidth}
        \centering
        \includegraphics[width=\textwidth]{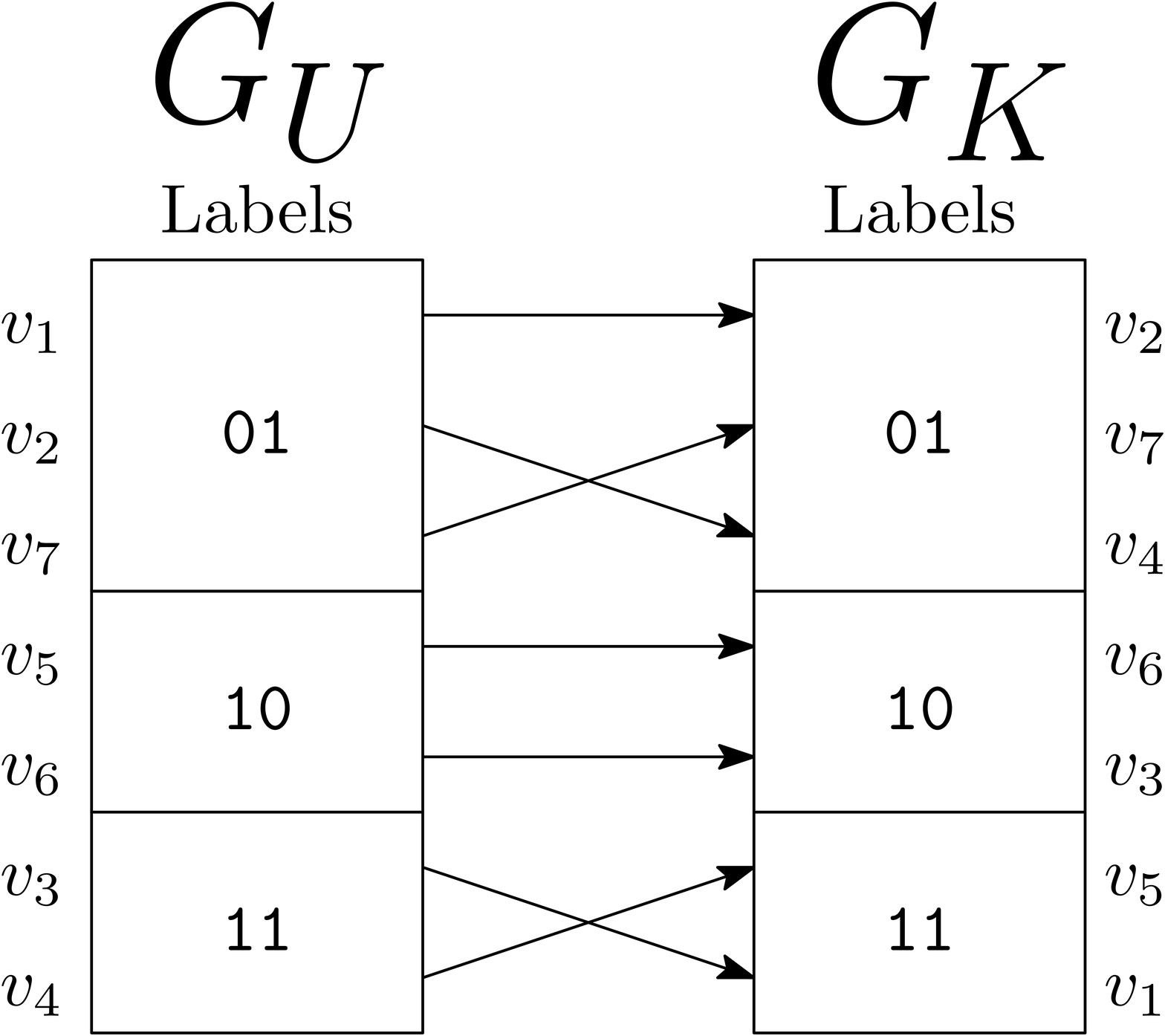}
        \caption{}
        \label{fig:mapping}
    \end{subfigure}
    \caption{
The unknown $G_U$ and known $G_K$ graphs are depicted in Sub-figure~\ref{fig:Gu}, and~\ref{fig:Gk}, respectively. The $C$ sequence consists of the nodes $c_1=v_6$ and $c_2=v_7$. Hence, the label of a node $v$ is a binary string $\ell_C^{G_U}=(b_2,b_1)\in\{0,1\}^2$, where $b_i=1 \Leftrightarrow c_i \in N(v)$, and the same for $G_K$. In the figures, we removed the super (sub) scripts as they are clear from the context. The induced labels are depicted next to the node number. Note that all the nodes have a neighbor in $C$, hence the label class ${\tt 00}$ is empty. In Sub-figure~\ref{fig:mapping} two tables are depicted: (1)~label assignment to the nodes of $G_U$, and (2)~label assignment to the nodes of $G_K$. The order of which the nodes are presented is according to the (single) isomorphism function between the two graphs. One can observe that the number of nodes per each corresponding label-classes, e.g., a buckets, is the same, as well as buckets which have the same ${\rm msb}$ bit, e.g., a cluster of labels. A random bijection that preserves the bucket of labels is depicted by the arrows. We show in Thm.~\ref{thm:findiso} that such a bijection yields graphs which are close to being isomorphic.}
\label{fig:labels}
\end{figure} 
\fi

\subsection{Distributed Algorithm Description}
The listing of the distributed algorithm appears in Algorithm~\ref{alg:testing}. The detailed description of the distributed implementation of Algorithm~\ref{alg:testing} appears in Section~\ref{sec:detailed}.

\begin{algorithm}[ht]
\DontPrintSemicolon
  \KwInput{A``known'' graph, $G_K = (V_K,E_K)$, that is an input to a single node $r$ (may be chosen adversarially).}
  \KwOutput{with high probability, all nodes output \yes\ if $G_K$ is isomorphic to $G_U$ and \no\ otherwise.}
\medskip
Compute a BFS tree, $T$, in $G_U$ rooted at $r$.\label{step0}\;
Pick, u.a.r., a sequence of $s\triangleq \Theta(\eps^{-1}\log |V_U|)$ nodes in $G_U$. Let $C\triangleq (c_1,\ldots,c_{s})$ denote this sequence.\label{step1}\;
Each node $v\in V(G_U)$ computes its label, $\ell^{G_U}_C(v)\in \{0,1\}^s$, according to $C$ and its neighbors in $G_U$ (see Definition~\ref{def1}), and sends this label to its neighbors. \label{stepeach}\;
%
%
The node $r$ picks a sequence of $t\triangleq \Theta(\eps^{-1}\log(|V_K|^{s}))$ pairs of nodes, $A = ((i_1, j_1),\ldots,(i_t, j_t))$, u.a.r. from $V_U\times V_U$.
Let $I \eqdef \{i_1, j_1, \ldots, i_t, j_t\}$. \;
For every $e \in A$, $r$ sends $e$ down the BFS tree and learns whether $e \in E_U$ or not.
Similarly, for every $v\in I$, it learns $\ell^{G_U}_C(v)$, i.e., the $C$-label of $v$ in $G_U$.\;
For each sequence, $P = (p_1, \ldots, p_s)$, of $s$ nodes from $V_K$, $r$ proceeds as follows:\;\label{testp}
\begin{enumerate}
\item For every $i\in [s]$, verify that $\ell^{G_U}_C(c_i) = \ell^{G_K}_P(p_i)$. If not, then \textbf{reject} $P$ as a candidate and proceed to the next sequence.\label{testp-another}\;
\item For every $v\in I$, let $\ell$ denote $\ell^{G_U}_C(v)$. Check if $|S^{G_U}_C(\ell)| =  |S^{G_K}_{P}(\ell)|$.\\
    If not, then \textbf{reject} $P$ as a candidate and proceed to the next sequence.\label{testp-1}\;
\item Pick uniformly at random a function $f$ from the set of all functions that are \\maximally $(C, P)$-label-consistent (see Definition~\ref{def3}).\label{testp-2}\;
\item Compute the number of edges in $A$ which are non-edges in $f(A)$ and vice-versa. \\That is, the number of edges $(i_k, j_k) \in A$ such that $(i_k, j_k) \in E_U$ and \\$(f(i_k), f(j_k)) \notin E_K$, or $(i_k, j_k) \notin E_U$ and $(f(i_k), f(j_k)) \in E_K$). \\If it is at most $(3\eps) |A|/2$ then \textbf{return}  \yes.\label{step-sym}\;
\end{enumerate}
\Indm
If all sequences, $P$, failed to pass the previous step then \textbf{return} \no.
%
%
\caption{Testing Isomorphism: The distributed network is $G_U=(V_U,E_U)$.}\label{alg:testing}
\end{algorithm}
\subsection{Correctness of the Distributed Testing Algorithm}

In this subsection we prove the correctness of our algorithm. We begin with a couple of claims and lemmas that we use in our proof. Missing proofs are deferred to Appendix~\ref{sec:mpfs}.

\medskip
\noindent
The proof of the following claim appears in~\cite{fischer2008testing}. The proof of Lemma~\ref{lem-close} can be derived from the proof of Lemma 4.11 in~\cite{fischer2008testing} 
(for the sake of completeness we provide both proofs in the appendix).

\begin{claim}[\cite{fischer2008testing}]\label{clm-sep}
For $\beta \in (0,1]$ and a sequence, $C$, of $s = \Theta(\log(n/\delta)/\beta)$ nodes, chosen uniformly
at random, $C$ is $\beta$-separating with probability at least $1 - \delta$.
\end{claim}
\newcommand{\cseven}{
\begin{proof}
Let $u$ and $v$ be such that $|\triangle(N(u), N(v))| \geq \beta n$. The probability that $C$ does not contain a node from $|\triangle(N(u), N(v))|$ is at most $(1-\beta)^s \leq \delta/n^2$, for an appropriate setting of $s$.
Therefore, the claim follows by a union bound over all pairs $u, v$.
\end{proof}
}
\begin{lemma}[\cite{fischer2008testing}]\label{lem-close}
Let $G$ and $H$ be isomorphic graphs and let $\pi$ be an isomorphism between them. For any $C$ that is an $\eps$-separating sequence of nodes of $G$ and for any $f$ that is $(C, \pi(C))$-label-consistent it holds that $\Delta(f(G), H) \leq \eps n^2$.
\end{lemma}
\newcommand{\leight}{
\begin{proof}
Let $C$ be an $\eps$-separating sequence of the nodes of $G$.
Since $\pi$ is an isomorphism between $G$ are $H$, it follows that $\pi(C)$ is $\eps$-separating in $H$ (since for any $u, v$ that have the same $\pi(C)$-label in $H$ it must hold that $\pi^{-1}(u), \pi^{-1}(v)$ have the same $C$-label in $G$).
Thus, by definition, for any pair of nodes $v, u$ such that $\ell_{\pi(C)}^{H}(v) = \ell_{\pi(C)}^{H}(u)$ it holds that $|\triangle(N_H(u), N_H(v))| < \eps n$.
Let $f$ be $(C, \pi(C))$-label-consistent.
By definition, $\ell^{G}_{C}(v) = \ell^{H}_{\pi(C)}(f(v))$, for every $v \in V(G)$.
Therefore, there exists a bijection $g: V(H) \rightarrow V(H)$ such that $f = g \circ\pi$ and $g$ only maps between nodes with the same $\pi(C)$-label.
Thus, $f$ can be obtain from $\pi$ by making at most $n$ swaps, one by one, between elements of $V(H)$ that have the same $\pi(C)$-label.
Since each swap changes the adjacency matrix by at most $\eps n$, we obtain the desired result.
\end{proof}
}
The following claim is implied directly from the multiplicative Chernoff's bound (see Theorem~\ref{thm-cher} in Section~\ref{sec-prob}).
\begin{claim}\label{lemma8}
Let $G$ and $H$ be two graphs such that $V(G) = V(H)$.  Then by querying the adjacency-matrices of $G$ and $H$ in $\Theta(\log(1/\delta)/\eps)$ random entries it is possible to distinguish between the case that $\Delta(G, H) > \eps n^2$ from the case that $\Delta(G, H) \leq \eps n^2/2$ with probability at least $1-\delta$.
\end{claim}
\newcommand{\cnine}{
\begin{proof}
Consider the outcome of querying the adjacency-matrices of $G$ and $H$ in $y \triangleq \Theta(\log(1/\delta)/\eps)$ random locations. Define the random variables $\{x_i\}_{i\in [y]}$ as follows: $x_i = 1$ if the values in the $i$-th location of both matrices are the same, and $0$ otherwise.
Let $p \triangleq \Delta(G, H)/n^2$ and define $\hat{p} = \sum_{i=1}^y x_i/|y|$. If $p= \eps/2$, then by Equation~\ref{eq-cher}, the probability that $\hat{p} > (3\eps)/2$ is at most $\delta$ (for the right setting of the parameter in the $\Theta$-notation).
Clearly, the same is true if $p< \eps/2$. On the other hand, if $p > \eps$ then by Equation~\ref{eq-cher} the probability that $\hat{p} \leq (3\eps)/2$ is at most $\delta$.
Therefore, by accepting if and only if $\hat{p} \leq (3\eps)/2$ we can distinguish $\Delta(G, H) = p \cdot n^2 >  \eps n^2$ from $\Delta(G, H) \leq \eps n^2/2$, as desired.
\end{proof}
}
\begin{lemma}\label{lemma:accepts}
If $G_U$ is isomorphic to $G_K$ then Algorithm~\ref{alg:testing} accepts with high probability.
\end{lemma}
\begin{proof}
Assume $G_U$ is isomorphic to $G_K$ and let $\pi: V(G_U) \rightarrow V(G_K)$ denote an isomorphism from $G_U$ to $G_K$.
Since the algorithm goes over every sequence $P$ of $s$ nodes from $V_K$, it also checks $\pi(C)$.
By Observation~\ref{lemma:same}, the probability that $\pi(C)$ passes Step~\ref{testp-1} is $1$.
By Claim~\ref{clm-sep}, with high probability, $C$ is $(\eps/2)$-separating (see Definition~\ref{def-sep}).
By Lemma~\ref{lem-close}, if $C$ is $(\eps/2)$-separating, then $f(G_U)$ is $(\epsilon/2)$-close to $\pi(G_U) = G_K$.
If $f(G_U)$ is $(\epsilon/2)$-close to $G_K$, then by Claim~\ref{lemma8}, with high probability $\pi(C)$ passes Step~\ref{step-sym}.
Therefore by the union bound, the algorithm accepts with high probability.
\end{proof}

\begin{lemma}\label{lemma:reject}
If $G_U$ is $\eps$-far from being isomorphic to $G_K$ then Algorithm~\ref{alg:testing} rejects with high probability.
\end{lemma}
\begin{proof}
Assume $G_U$ is $\eps$-far from being isomorphic to $G_K$. We claim that with high probability, any sequence $P$, fails to pass Step~\ref{testp}.
We show this by bounding the probability that a fixed $P$ passes Step~\ref{testp} and then apply the union bound over all possible sequences.
Fix a sequence $P$ and assume that $P$ passes Step~\ref{testp-1} (otherwise we are done).
Let $f$ be the corresponding $f$ from Step~\ref{testp-2} (that is chosen at random).
Since $G_U$ is $\eps$-far from being isomorphic to $G_K$, by definition, $\Delta(f(G_U), G_K) \geq \eps n^2$.
Recall that $f$ is chosen uniformly at random from the set of functions that are maximally $(C, P)$-label-consistent.
It is not hard to verify that $f$ and $A$ are independent random variables.
Therefore we can apply Claim~\ref{lemma8} on Step~\ref{step-sym} as $A$ is a set of potential edges chosen uniformly at random and, in particular, independently from $f$.
By Claim~\ref{lemma8}, $P$ succeeds to pass Step~\ref{step-sym} with probability at most $1/|V_K|^{s+c}$ for any absolute constant $c$. Thus, the lemma follows by a union bound over all possible sequences, as their number is bounded by $|V_K|^s$.
\end{proof}

\subsection{A Detailed Description of the Distributed Implementation of Algorithm~\ref{alg:testing}}\label{sec:detailed}
In this section we provide a detailed description of the distributed implementation of our algorithm. We focus on steps for which the implementation is not straightforward and analyze their round complexity. In particular, we focus on Steps~\ref{step1},6.\ref{testp-1},6.\ref{testp-2}-6.\ref{step-sym}.
\subparagraph*{Step~\ref{step1}: Selecting $s$ Nodes u.a.r.}
We propose the following simple procedure to select $s$ nodes uniformly at random (which is a kind of folklore).
Each node selects a random number in $[n^{c+2}]$ where $c$ is an absolute constant.
For a fixed pair of nodes, the probability that both nodes pick the same number is at most $1/n^{c+2}$.
Therefore, by union bound over all pairs, with probability at least $1-1/n^c$, all selected numbers are distinct.
Conditioned on this event, the nodes with the $s$ highest numbers are distributed uniformly at random.
Each node sends its ID and its selected number up the BFS tree and the messages are forwarded up the tree in a manner that prioritizes messages whose number is higher.
Therefore, the root receives the $s$ highest numbers (along with the IDs of the corresponding nodes) in $D+s$ rounds.
To see this observe that the message with the highest number is never delayed and in general the message with the $i$-highest number may be delayed for at most $i-1$ rounds.

\subparagraph*{Step~6.\ref{testp-1}: Computing $S_C^{G_U}(\ell)$.}
Clearly, $r$ can compute $S_P^{G_K}(\ell)$ for any label $\ell$ as $r$ knows $P$ and $G_K$. Therefore, in order to describe the implementation of Step~\ref{testp-1} it suffices to explain how $r$ can obtain $|S_C^{G_U}(\ell)|$.

We begin with the special case of $\ell = (0, \ldots, 0)$. The nodes in $G_U$ that have this label are nodes that are not adjacent to any one of the nodes in $C$.
Their number can be computed in $O(D)$ rounds by summing it up the BFS tree as follows. Assume w.l.o.g. that every node knows its layer in the BFS-tree. In the first round, every node that is in the last layer (which is also a leaf) sends $1$ up to its parent.
In the next round, all nodes in the next layer sum up the received numbers and add $1$ if their $C$-label is $(0, \ldots, 0)$. They send this number up to their parents and so on until we get to the root.

Consider a label $\ell$ for which at least one bit it $1$. Let ${\rm msb}(\ell)$ denote the maximum $i$ such that $\ell_i = 1$. Since the node $c_{{\rm msb}(\ell)}$ is connected to all nodes whose $C$-label is $\ell$, it can compute their total number (recall that in Step~\ref{stepeach} every node sends its $C$-label to all its neighbors) and send it to the root.
Therefore, by a pipelining argument,  the root can obtain $|S_C^{G_U}(\ell)|$ for every $\ell$ which is a $C$-label of a node in $I$ in $O(D+|I|) = O(D+ (\eps^{-1}\log n)^2)$ rounds.

\subparagraph*{Steps~6.\ref{testp-2}-6.\ref{step-sym}: Accessing $f$.}
Recall that we require from $f$ to be chosen u.a.r. from the set of all functions that are maximally $(C, P)$-label-consistent (see Definition~\ref{def3}).
Recall that $f$ is only evaluated on nodes in $I$ but at the same time its selection has to be independent of $A$ (and $I$).
To this end, the root verifies the following:
\begin{enumerate}
\item In Sub-step~\ref{testp-another} of Step~\ref{testp} it verifies that for the selected sequences, $C$ and $P$, corresponding nodes have matching labels. Namely, $\ell^{G_U}_C(c_i) = \ell^{G_K}_P(p_i)$ for every $i\in [s]$.
\item In Sub-step~\ref{testp-1} of Step~\ref{testp} it verifies that $|S^{G_U}_C(\ell)| =  |S^{G_K}_{P}(\ell)|$ for every $\ell$ which is a $C$-label of a vertex in $I$.
\end{enumerate}
If both conditions hold, then it follows that by mapping the nodes in $S^{G_U}_C(\ell)$ to the nodes in $S^{G_K}_{P}(\ell)$ u.a.r. and independently from the mapping of all other nodes (except for the mapping of $C$ to $P$ which is already determined) for every $\ell$ such that $|S^{G_U}_C(\ell)| =  |S^{G_K}_{P}(\ell)|$, we are in fact accessing $f$ which is drawn according to the desired distribution.
Therefore the root $r$ simply maps every $v\in I$ to a uniform node $u\in V_K$ such that: (1) $\ell^{G_U}_C(v) = \ell^{G_K}_{P}(u)$ (2) $u$ is still unmapped (such node always exists). Since $r$ knows $G_K$ and $\ell^{G_U}_C(v)$ for every $v\in I$, it is able to compute $f(v)$ for every $v \in I$, as desired.

\def\wG{{\widetilde{G}}}

%

\section{Computing an Approximated Isomorphism}\label{sec:isoapx}
In this section we prove the following theorem.



\begin{theorem}\label{thm:findiso}
Let $G_U$ denote the input graph and let $G_K$ be a graph which is isomorphic to $G_U$ and is given as an input to all nodes in the network.
There exists a randomized algorithm such that each node in $G_U$, $v$, outputs $g(v)$ where $g: V_U \rightarrow V_K$ is a bijection such that $g(G_U)$ is $\epsilon$-close to $G_K$.
The round complexity of the algorithm is $O(D+(\eps^{-1}\log n)^2)$. The algorithm succeeds with high probability.
\end{theorem}

\begin{proof}
The first step of the algorithm is to run Algorithm~\ref{alg:testing} with the only difference that in Step~\ref{testp} the root also verifies that $|S^{G_U}_C(\ell)| =  |S^{G_K}_{P}(\ell)|$ for $\ell = (0, \ldots, 0)$.
Since $G_K$ and $G_U$ are isomorphic, by Lemma~\ref{lemma:accepts}, w.h.p. the algorithm accepts and hence finds $P$ and the corresponding $f$ that pass Step~\ref{testp}.
Recall that w.h.p. $f(G_U)$ is $\eps$-close to $G_K$.
If every node $v$ could output $f(v)$ then we were done.
However, we can not compute $f$ for every node $v$ because for a constant fraction of the nodes its computation might require global information on $G_U$.
Instead, our goal is to output $g$ which is $O(\epsilon)$-close to $f$ and can be computed for every node without causing too much congestion. We next describe $g$ and its computation.
%

We begin with some notation.
Let $L_i \subseteq \{0,1\}^s$ denote the set of labels $\ell$ for which ${\rm msb}(\ell) = i$.
Let $Y$ denote the set of nodes in $V_U$ whose $C$-label is $(0, \ldots, 0)$.
Similarly, let $Y'$ denote the set of nodes in $V_K$ whose $P$-label is $(0, \ldots, 0)$.
For a graph $H$ and a sequence $D$ let $J^H_D(i) \eqdef  \cup_{\ell \in L_i} S^H_D(\ell)$, namely, this is the set of all nodes in $H$ whose $D$-label belongs to $L_i$.
We may refer to $J^H_D(i)$ as the $i$-th cluster of the graph $H$ w.r.t. $D$.
For $i\in s$, define $j_i \eqdef |J^{G_U}_C(i)| - |J^{G_K}_{P}(i)|$.
Namely, $j_i$ is the difference between the sizes of the $i$-th clusters in both graphs (w.r.t. $C$ and $P$, respectively).

We next define the set of {\em reserved} nodes of $V_K$, denote by $R$.
For each $i$ such that $j_i < 0$, $|j_i|$ nodes from $J^{G_K}_{P}(i)$ belong to $R$.
Specifically, these are the nodes whose order~\footnote{We assume that there is a total order on $V_K$ which is known to all the nodes in $G_U$.} is the least from the vertices in $J^{G_K}_{P}(i)$.
We consider the order to be the same order as in $V_K$ only that elements in $P$ have the highest order (this is to ensure that none of the elements in $P$ belong to $R$).

We are now ready to describe $g$. Let $v\in V_U$ and let $\ell =\ell^{G_U}_C(v)$.
We assume that $v \notin Y$ as we explain the mapping of the nodes in $Y$ separately.
We consider the following cases.

{\bf The first case is when $j_i = 0$.} We have the following sub-cases.
\begin{enumerate}
\item For every $\ell' \in L_i$ such that $|S^{G_U}_C(\ell')| = |S^{G_K}_{P}(\ell')|$, $g$ matches u.a.r. the elements in $S^{G_U}_C(\ell')$ to the elements in $S^{G_K}_{P}(\ell')$.
\item The rest of the elements in $J^{G_U}_C(i)$ are matched u.a.r. to the unmatched elements in $J^{G_K}_{P}(i)$.
\end{enumerate}
Therefore, in this case the elements in $J^{G_U}_C(i)$ are matched only to the elements in $J^{G_K}_{P}(i)$ and vice versa.

{\bf The second case is when $j_i < 0$.} We have the following sub-cases.
\begin{enumerate}
\item For every $\ell' \in L_i$ such that $|S^{G_U}_C(\ell')| = |S^{G_K}_{P}(\ell')|$ and $S^{G_K}_{P}(\ell') \cap R = \emptyset$, $g$ matches the elements in $S^{G_U}_C(\ell')$ u.a.r. to the elements in $S^{G_K}_{P}(\ell')$.
\item For every $\ell' \in L_i$ such that $|S^{G_U}_C(\ell')| = |S^{G_K}_{P}(\ell')|$ and $S^{G_K}_{P}(\ell') \cap R \neq \emptyset$, $g$ matches u.a.r. a random set of $|S^{G_K}_{P}(\ell') \setminus R|$ elements from $S^{G_U}_C(\ell')$  to $S^{G_K}_{P}(\ell') \setminus R$.
\item The rest of the un-matched elements in $J^{G_U}_C(i)$ are mapped u.a.r. to the un-matched elements in $J^{G_K}_{P}(i) \setminus R$.
\end{enumerate}
Observe that all the elements in $J^{G_U}_C(i)$ are matched to elements in $J^{G_K}_{P}(i)$ and that the elements that belong to $J^{G_K}_{P}(i) \cap R$ are still un-matched.

{\bf The third case is when $j_i > 0$.} We have the following sub-cases.
\begin{enumerate}
\item For every $\ell' \in L_i$ such that $|S^{G_U}_C(\ell')| = |S^{G_K}_{P}(\ell')|$, $g$ matches the elements in $S^{G_U}_C(\ell')$ u.a.r. to the elements in $S^{G_K}_{P}(\ell')$.
\item The rest of the elements in $J^{G_U}_C(i)$ are matched u.a.r. to the unmatched elements in $J^{G_K}_{P}(i)$.
\item The remaining $|j_i|$ elements in $J^{G_U}_C(i)$ are matched to the nodes of order $(\sum_{a < i} j_a) + 1$ to $(\sum_{a < i} j_a) + j_i$ in $R$.
\end{enumerate}

This concludes the description of $g$ for nodes that do not belong to $Y$.
Before we explain how $Y$ is matched to $Y'$ we first describe how $g$ can be computed distributively for nodes that have at least one neighbor in $C$ (namely, nodes that do not belong to $Y$).
Each node $c_i \in C$ is responsible to compute and to send to each node $v$ whose $C$-label is in $L_i$ the value $g(v)$ (note that $c_i$ and $v$ are necessarily neighbors).
As a preliminary step, each node $c_i$ computes $j_i = |J^{G_U}_C(i)| - |J^{G_K}_{P}(i)|$ and sends $(i, j_i)$ up the BFS tree.
Notice that $\sum_{i\in s} j_i = 0$ as $|S^{G_U}_C(\ell)| =  |S^{G_K}_{P}(\ell)|$ for $\ell = (0, \ldots, 0)$ and $|V_U| = |V_K|$.
The root sends the set $\{(i, j_i)\}_{i\in s}$ down the BFS tree.
By knowing $G_K$ and the set $\{(i, j_i)\}_{i\in s}$, every node $c_i$ can easily compute $R$.
It is not hard to see that this suffices in order to match the elements in $J^{G_U}_C(i)$ to $V_K$ as described above.

We next describe the matching of $Y$ to $Y'$ and explain how it is computed distributively.
We aim to assign to each node in $Y$ a label in $[|Y|]$ uniquely.
This way each node in $Y$ can match itself to a node in $Y'$ (recall that all nodes know $V_K$ and the total order on $V_K$).
To this end, we use the BFS tree as follows. Each node in the BFS tree computes how many nodes in its subgraph are in $Y$.
This can be done in $O(D)$ rounds as follows. We assume w.l.o.g. that each node knows its layer in the BFS tree.
Let $b$ denote the number of layers. We proceed in $b$ rounds.
In the first round, every node in $Y$ which is in the $b$-th layer sends to its parent the message $1$.
In the next round, all the nodes in layer $b-1$ sum up the messages they received and add $1$ if they belong to $Y$.
Then they send the result to their parents and so on until we end up at the root.
Now the root partitions the interval $[1, \ldots, |Y|]$ into consecutive sub-intervals and assigns these sub-intervals to its children.
Each child receives an interval whose size equals to the number of nodes in its subgraph that are in $Y$.
In a similar manner, these sub-intervals are partitioned recursively down the tree until each node in $Y$ is assigned with a unique number in $|Y|$, as desired.

By construction it follows that $g$ is a bijection.
The bound on the round complexity follows from the bound on the round complexity of Algorithm~\ref{alg:testing} and the fact that there are only $s = O(\eps^{-1}\log n)$ clusters.
\ifnum\icalp=1
It remains to prove the following claim, 
which is deferred to Appendix~\ref{sec:mpfs}.
\else
It remains to prove the following claim. 
\fi
\begin{claim}\label{clm:1}
With high probability $\Delta(g(G_U), G_K) < \eps n^2$.
\end{claim}
\newcommand{\cthirt}{
\begin{proof}
We observe that both $f$ and $g$ are random variables.
To prove the claim about $g$ we couple $g$ to $f$.
From the fact that w.h.p. $\Delta(f(G_U), G_K) < \eps n^2$ with combination with the coupling it will follow that w.h.p. $\Delta(g(G_U), G_K) = O(\eps n^2)$.
Therefore, by setting the proximity parameter to be $\Theta(\eps)$ the claim will follow.

Consider the following description of $g$ in terms of $f$.
Let $B = \{\ell \in \{0,1\}^s: |S^{G_U}_C(\ell)| \neq |S^{G_K}_{P}(\ell)|\}$.
For every $v \in V_U$, $g(v) = f(v)$ unless $v \in Y \cup B$ or, $v \notin Y \cup B$ and $f(v) \in R$.
In these cases we match $v$ to a node in $V_K$ as described above.
Observe that under this formulation the distribution of $g$ remains the same.
The only difference is that now, for the sake of the analysis, it is coupled to the distribution of $f$.

It follows that the number of nodes $v\in V_U$ for which $f(v) \neq g(v)$ is at most $|Y| + |R| + |B|$.
By Step~\ref{testp-1} of Algorithm~\ref{alg:testing}, with high probability it holds that $\sum_{\ell \in B} |S^{G_U}_C(\ell)| \leq \eps n$.
On the other hand, w.h.p., the number of neighbors of every node in $Y$ is at most $\eps n$.
This implies that the number of nodes in the neighborhood of $Y'$ is at most $O(\eps n^2)$ (since otherwise would reject $f$ w.h.p.).
Therefore, w.h.p., the contribution of the nodes in $Y$ and $Y'$ to $\Delta(f(G_U), g(G_U))$ is at most $O(\eps n^2)$.
Thus, w.h.p. $\Delta(f(G_U), g(G_U)) = O(\eps n^2)$, as desired.
By Claim~\ref{clm:1}, w.h.p. $\Delta(f(G_U), G_K) < \eps n^2$, hence we obtain by the union bound that w.h.p. $\Delta(g(G_U), G_K) < O(\eps n^2)$. By setting the proximity parameter appropriately we obtain the desired result.
\end{proof}
}
\ifnum\icalp=0
\cthirt
\fi

This concludes the proof of the Theorem.

\end{proof}

\ifnum\icalp=0
\section{Distributed Algorithm for Deciding Isomorphism}\label{sec:exact}

In this section we prove the following theorem.

\begin{theorem}\label{thm:exact}
  There exists a randomized distributed algorithm that decides if $G_K$ and $G_U$ are isomorphic with high probability. The round complexity of the algorithm is $O(n)$.
\end{theorem}

\begin{proof}
The idea of the algorithm is to go over all possible one-to-one mappings between the nodes of $G_K$ and the nodes of $G_U$ and to test for equality of the corresponding adjacency matrices.
The test for equality is performed with very high confidence level in order to ensure that the total error probability is bounded by a small constant.
We note that a similar reduction for testing equality also appears in the algorithm of~\cite[Sec.~6.2]{DBLP:journals/corr/abs-1901-01630} for the Identical Subgraph Detection problem (ISDP).

The first step of our algorithm is to construct a BFS tree and to assign to each node in the network a unique label in $[n]$ where $n \eqdef |V(G_K)|$.
This step requires $O(D)$ rounds where $D$ denotes the diameter of $G_U$ (see details on the implementation of this step in the proof of Theorem~\ref{thm:findiso}).
Consider the adjacency matrix of $G_U$, $M$, in which the rows are sorted according to the labels assigned to the nodes.
We consider the natural total order on the following set of pairs of nodes $P = \{(i, j) : i\in [n], j\in [n], i< j\}$ in which the pairs are sorted according to the first element and ties are broken according to the second element. Let $\ell(i, j)$ denote the order of the pair $(i, j)$.
Each pair $(i , j)\in P$ corresponds to the potential edge between the pair of nodes with labels $i$ and $j$, respectively.
The matrix $M$ can be represented as an integer $s(M)$ where for each $(i, j) \in P$ the $\ell(i, j)$-th lsb (least significant bit) of the binary representation of $s(M)$ indicates whether $(i, j)$ is an edge in the graph.
Observe that we can calculate $s(M)$ in $D$ rounds by starting the calculation at the lowest layer of the BFS tree and summing up the outcomes as we go up the tree, layer by layer.
The calculation is performed such that each power of two, $2^{\ell}$, is added (once) if and only if the corresponding edge is present in the graph (i.e. $2^{\ell(i, j)}$ is added if and only if $(i, j)$ is present in the graph).

Since the representation of $s(M)$ requires $O(n^2)$ bits our goal we calculate $s(M) \mod p$ instead.
To this end we proceed in the same manner as mentioned above only that before the nodes send up the tree the outcome of the intermediate sums, they apply the $\mod p$ operation on the outcome.
Let $\mathcal{P} = \{p_1,\ldots, p_k\}$ be a multiset of $k$ prime numbers, where $k = O(n)$, each chosen independently and uniformly at random from the first $n^2$ primes.~\footnote{It is well known that for sufficiency large number $x \in \NN$ the number of prime numbers that are at most $x$ is $\Theta(x/\log x)$.}
Let $M'$ denote the adjacency matrix of $G_K$.
If $M' \neq M$ then the probability that $s(M) \equiv s(M') \mod p$ for a random prime number in $[n^2]$ is at most $1/\Omega(n)$~\cite{DBLP:books/daglib/0011756}.
Therefore the probability that $s(M) \equiv s(M') \mod p$ every $p\in \mathcal{P}$ is at most $1/\Omega(n^k) = 1/\Omega(n^n)$.
Thus we can test with one-sided error if $G_U$ and $G_K$ are equal. The soundness of the equality test is $1 - 1/\Omega(n^n)$.
We can apply the equality test for every mapping $\pi$ between $G_K$ and $G_U$ and return \yes\ if and only if there exists a mapping for which the test accepts.
Namely, we go over all possible permutations over the nodes of $G_K$ and for each permutation, $\pi$, we perform the equality test between $s(M)$ and $s(\pi(M'))$ where $\pi(M')$ denotes adjacency matrix of $G_K$ after applying the permutation $\pi$ on the nodes.
By the soundness of the equality test and the union bound, the probability that this test returns \no\  when $G_K$ and $G_U$ are not isomorphic is at least $2/3$ (for an appropriate adjustment of the parameters).

By standard pipelining, it is possible to calculate $s(M) \mod p$ for every $p \in \mathcal{P}$ in $O(D + k)$ rounds, therefore the round complexity of the above test is $O(n)$.
To verify this observe that in order to execute the above test the only information we need is of $G_K$ and  the result of $s(M) \mod p$ for every $p \in \mathcal{P}$.
This concludes the proof.
\end{proof}

\medskip\noindent
We observe that the above algorithm can be adapted to the semi-streaming model in a straight-forward way as follows.

\begin{theorem}\label{thm:sstream}
There exists an algorithm in the semi-streaming model that receives a graph $G_K$ over $n$ nodes as an input, where the space for storing $G_K$ is a read-only memory,
and a stream of the edges of another graph $G_U$ (according to any order) and decides, with one-sided error, whether $G_K$ and $G_U$ are isomorphic or not. The algorithm performs one-pass and uses $O(n\log n)$ bits of space.
\end{theorem}
\begin{proof}
Assume w.l.o.g. that the labels of the nodes in $G_U$ and $G_K$ are taken from $[n]$.
Otherwise we can re-name then by using a table of size $O(n\log n)$ bits.
Let $M$ denote the adjacency matrix of $G_U$.
We first compute $s(M) \mod p$ for every $p \in \mathcal{P}$ as in the proof of Theorem~\ref{thm:exact}, in one-pass, using $O(n\log n)$ bits of space.
We then go over all permutations of the nodes of $G_K$ and perform the same computation for the corresponding adjacency matrix.
We accept if and only if there exists a permutation $\pi$ for which $s(M) \equiv s(M') \mod p$ every $p\in \mathcal{P}$, where $M'$ denotes the adjacency matrix of $G_K$ after we permuted the nodes according to $\pi$.
Observe that we can go over the permutations one by one according to the lexicographical order by using $O(n\log n)$ bits of space.
Therefore, the total space the algorithm uses is $O(n\log n)$, as desired.
\end{proof}

\def\wG{{\widetilde{G}}}

\section{Lower Bounds}\label{sec:lb}
In this section we establish two lower bounds. The first is for the decision variant of \GI, and the second is for the testing variant.

For the decision variant, we prove a near-quadratic lower bound for any deterministic distributed algorithm in the \congest\ model. The second lower bound states that any Isomorphism testing distributed algorithm requires diameter time. This lower bound holds also for randomized algorithms, and in fact holds in the \local\ model, even if \emph{all} vertices are given as an input the graph $G_K$.\footnote{In the \local\ model there is no limitation on message size per round per edge. Obviously, a lower bound in the \local\ model also applies to the \congest\ model.}
\subsection{An $\Omega(n^2/\log n)$ Lower Bound for Deciding Isomorphism Deterministically}\label{sec:detlb}

The decision variant of our Isomorphism testing problem is as follows: if $G_K$ and $G_U$ are isomorphic, then all nodes should output \yes, while if the graphs are not isomorphic, at least one node should output \no.

\begin{theorem}\label{thm:detlb}
Any distributed deterministic algorithm in the \congest\ model for deciding whether $G_U$ is isomorphic to $G_K$
requires $\Omega(n^2 / \log n)$ rounds.
\end{theorem}
\begin{proof}
We reduce the problem of Set-Equality to the problem of deciding Isomorphism.
 The reduction is as follows. Alice and Bob each receives as an input a subset of $k^2$ elements, $A\subseteq\{(x,y) \mid x\in \{a_1^0,\ldots, a_1^{k-1}\}, y \in \{a_2^0,\ldots, a_2^{k-1}\}\}$ and $B\subseteq\{(x,y) \mid x\in \{b_1^0,\ldots, b_1^{k-1}\}, y \in \{b_2^0,\ldots, b_2^{k-1}\}\}$, respectively. According to their inputs they construct a pair of graphs, $G_U$ and $G_K$, such  that $A=B$ if and only if $G_U$ is isomorphic to $G_K$, as described momentarily.
Since, Set-Equality has a deterministic communication complexity which is linear in size of the universe~\cite{DBLP:books/daglib/0011756}, which is $\Omega(k^2)$ in this case, it follows that any distributed deterministic algorithm in the \congest\ model for deciding whether $G_U$ is isomorphic to $G_K$ requires $\Omega(n^2 / \log n)$ rounds.

We now describe the graph $G_U$ which is constructed by Alice and Bob.
The set of nodes in Alice's graph is composed of $u, u', u''$, the subsets $A_1 = \{a_1^0,\ldots, a_1^{k-1}\}$ and $A_2 = \{a_2^0,\ldots, a_2^{k-1}\}$, and the nodes $t_{A_1}, t_{A_2}, t_{u}^1, t_{u}^2$ (see Figure~\ref{fig:detlb}).
The subgraph induced on the nodes in $A_1$ is the path $a_1^0,\ldots, a_1^{k-1}$.
Similarly, the subgraph induced on the nodes in $A_2$ is the path $a_2^0,\ldots, a_2^{k-1}$.
The nodes $u$ and $u''$ are adjacent to all nodes in $A_1$ and $A_2$, respectively.
Both nodes $u$ and $u''$ are adjacent to the node $u'$.
The node $u$ is adjacent to $t_u^1$ and $t_u^2$.
The nodes $a_1^0$ and $a_2^0$ are adjacent to the nodes $t_{A_1}$ and $t_{A_2}$, respectively.
Additionally, the edge $(a^x_1, a^y_2)$ belongs to Alice's subgraph if and only if the corresponding $(x,y)$ element is in $A$.

The subgraph of Bob is defined similarly only that the nodes $v, v', v'', t_{B_1}, t_{B_2}$ take the role of $u, u', u'', t_{A_1}, t_{A_2}$, respectively, and
the subsets $B_1$ and $B_2$ take the role of $A_1$ and $A_2$, respectively.
The edges between $B_1$ and $B_2$ are determined by Bob's input $B$.
The other difference is that $v$ is adjacent to three ``tails'' $t_v^1, t_v^2, t_v^3$ (whereas $u$ is adjacent to only two, $t_v^1, t_v^2$).

The subgraphs of Alice and Bob are connected by a single edge $(u, v)$.

The graph $G_K$ is constructed by Alice exactly as the graph $G_U$, the only difference is that Alice does not know the subset $B$, instead of $B$ Alice uses $A$ to determine both the edges between $A_1$ and $A_2$ and the edges between $B_1$ and $B_2$.

Clearly, if $A$ and $B$ are such that $(a_1^i, a_2^j) \in A$ if and only if $(b_1^i, b_2^j) \in B$, then $G_U$ and $G_K$ are isomorphic.
We next prove that if this is not the case then $G_U$ and $G_K$ are not isomorphic.
To show this we prove that given the structure of $G_U$, namely $G_U$ with arbitrary labels on the nodes, we are able to recover the subsets $A$ and $B$ (i.e. there is a one-to-one correspondence between $G_U$ and the subsets $A$ and $B$).
We first observe that it is possible to identify $u$ and $v$ based on the structure of $G_U$ because these are the only nodes that are neighbors to $2$ and $3$ nodes with degree exactly $1$, respectively.
Next, we identify $u'$ and $v'$ - these are the nodes that have degree exactly $2$ and are adjacent to $u$ and $v$, receptively.
Once we identify $u'$ and $v'$ we can identify $u''$ and $v''$ and in turn identify $A_2$ and $B_2$.
The nodes in $A_1$ and $B_1$ are the neighbors of $u$ and $v$, excluding $u'$ and $v'$, receptively.
Finally, we are able to identify $a_1^0, a_2^0, b_1^0$ and $b_2^0$ since these are the only nodes that are adjacent to a node of degree exactly one - $t_{A_1}, t_{A_2}, t_{B_1}$ and $t_{B_2}$, respectively.
Once we identify $a_1^0$ and $A_1$ we can recover the path $a_1^0,\ldots, a_1^{k-1}$.
Likewise for the paths induced on $A_2, B_1$ and $B_2$.
This implies that we can recover the subsets $A$ and $B$, which concludes the proof.
  \begin{figure}
\centering
\includegraphics[width=0.9\linewidth]{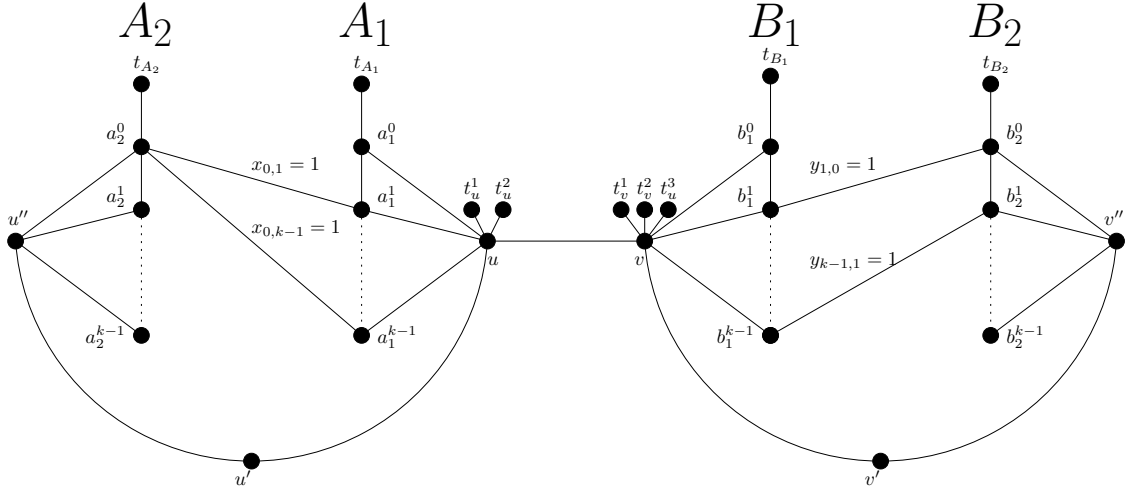}
\caption{Deciding Isomorphism lower bound construction.}
\label{fig:detlb}
\end{figure}
\end{proof}

\subsection{An $\Omega(D)$ lower bound for the Testing Isomorphism Problem}\label{sec:dlb}
In this section we establish the following theorem.
\begin{theorem}\label{thm:lb}
For any $D$ there exists a family of graphs with diameter $\Theta(D)$ such that any distributed two-sided error property testing algorithm for testing isomorphism on this family of graphs
requires $\Omega(D)$ rounds. This lower bound applies also in the \local\ model and even if all nodes receive $G_K$ as an input.
\end{theorem}

\begin{proof}
We define a family of graphs ${\mathcal H}$, where each $H \in {\mathcal H}$ takes the role of the unknown graph $G_U$. We fix $G_K \in {\mathcal H}$ and require from ${\mathcal H}$ that all graphs in ${\mathcal H} \setminus \{G_K\}$ are $\epsilon'$-far from being isomorphic to $G_K$.
Theorem~\ref{thm:lb} then follows by applying Yao's principle~\cite{yao1977probabilistic}.

\subparagraph*{Construction of ${\mathcal H}$.}
We begin our construction with a pair of graphs on $n$ nodes and diameter $\Theta(D)$, $G_1$ and $G_2$, such that $|E(G_1)| = (1+\eps)|E(G_2)|$.~\footnote{Observe that such graphs exist for any $D$ for $G_1$ and $G_2$ that are dense.
On the other extreme, if we want $G_1$ and $G_2$ to be bounded degree trees then it is possible to obtain such graphs for any $D = \Omega(\log n)$.}
Clearly, $G_2$ is $\epsilon$-far from being isomorphic to $G_1$ (since $G_1$ has $\eps |E(G_2)|$ more edges than $G_2$).
Let ${\mathcal H}$ denote a family of graphs where each graph in ${\mathcal H}$, denoted by $G_{i,j}$ for $i,j \in \{1,2\}$, is constructed by using two copies from $\{G_1,G_2\}$. The definition of $G_{i,j}$ is as follows:
\begin{enumerate}
	\item Let $p\triangleq(p_1,\ldots,p_{D})$ denote a path of $D$ nodes.
	\item Let $v_i$ and $v_j$ denote arbitrary nodes in $G_i$ and $G_j$, respectively.
	\item Identify, $v_i$ with $p_1$ and $v_j$ with  $p_{D}$, to obtain $G_{i,j}$. Observe that $ |V(G_{i,j})|=2n + D - 2 = O(n)$, and that the diameter of $G_{i,j}$ is $\Theta(D)$.
\end{enumerate}

Note that from the reasoning above, for every $G_{i,j}, G_{a,b}\in {\mathcal H}$, where $\{i,j\}\neq \{a,b\}$, $G_{i,j}$ is $\Omega(\eps)$-far from being isomorphic to $G_{a,b}$.

We assume towards a contradiction that there is a two-sided tester $\mathcal{A}$ that is correct with probability at least $2/3$ and runs for $\diam{G_U}/3$ rounds, where $\diam{G_U}$ denotes the diameter of the distributed network $G_U$.
The lower bound then follows from Yao's principle~\cite{yao1977probabilistic}.
We next specify the distribution over the inputs and explicitly apply Yao's principle.

We define a set of inputs to the distributed tester: (1)~we fix the known graph $G_K$ to be $G_{1,2}$, (2)~the unknown graph $G_U$ is chosen randomly from the distribution $\tilde{\cal{H}}$ which is defined as follows: the probability to obtain the graph $G_{1,2}$ is $1/2$ and each of the graphs $G_{1,1}$ and $G_{2,2}$ is obtained with probability $1/4$.
Now we consider two cases. The first case is that $\mathcal{A}$ outputs (correctly) \yes\ when $G_U$ is $G_{1,2}$.
By the definition of the problem this implies that all nodes output \yes.
Since the number of rounds of the tester is at most $D/3$, it follows that all nodes in both graphs $G_{1,1}$, and $G_{2,2}$, output \yes.
To verify this observe that for any node $u$ in either $G_{1,1}$ or $G_{2,2}$ there exists a node in $G_{1,2}$, $v$, such that the $(D/3)$-hop neighborhoods of $v$ and $u$ are the same (up to labels and port numbers~\footnote{We assume that the output of the algorithm is invariant to the labeling of the nodes and the port numbers of the edges. In the Appendix~\ref{app:liftassumplb} we remove this assumption and provide a general and more detailed proof.}).
Thus, in this case, $\mathcal{A}$ errs on $G_{1,1}$ and $G_{2,2}$.
In the second case $\mathcal{A}$ outputs \no\ when $G_U$ is $G_{1,2}$. Thus, in both cases $\mathcal{A}$ errs with probability at least $1/2$ when $G_U$ is drawn according to $\tilde{\cal H}$.

By Yao's principle it is implied that any randomized tester must err with probability at least $1/2$ as well, in contradiction to our assumption that $\mathcal{A}$ is correct with probability of at least $2/3$.

\end{proof} 
\section{Simulating Centralized Property Testing Algorithms}\label{sec:simulc}
In this section we prove the following claim.

\begin{claim}\label{clm:sim}
Let $\mathcal{A}$ be a centralized property testing algorithm that is allowed to make adjacency queries, incidence queries and degree queries.
There exists a distributed algorithm that can simulate $\mathcal{A}$ in $O(D \cdot q)$ rounds if $\mathcal{A}$ is adaptive and in $O(D + q)$ rounds if $\mathcal{A}$ is non-adaptive, where $D$ denotes the diameter of the graph and $q$ denotes the number of queries that $\mathcal{A}$ makes.
\end{claim}
\begin{proof}
In order to simulate $\mathcal{A}$ on the distributed network we first pick a leader $r$ and construct a BFS tree rooted at $r$ in $D$ rounds.
We also assume w.l.o.g. that $r$ knows the size of the network, $n$.
If $\mathcal{A}$ is non-adaptive then $r$ can determine the queries that $\mathcal{A}$ makes by using (only) the knowledge of $n$ and random bits.
The root $r$ sends this information to the entire network in $D + q$ rounds and in additional $D+ q$ rounds the answers are gathered by pipelining back at $r$.
Once that all the answers are gathered, $r$ can complete the simulation of $\mathcal{A}$. The outcome is then sent to all the nodes in the network.
If $\mathcal{A}$ is adaptive then $r$ simulates $\mathcal{A}$ step by step where each query is gathered in $2D$ rounds. This yields a round complexity of $O(D \cdot q)$, as desired.
\end{proof}

\subparagraph*{Application to minor-free graphs and Outerplanar graphs.}

For graphs which are minor-free with a bound $d=O(1)$ on the degree Newman and Sohler~\cite{NS13} argued the following in the centralized property testing model: (a)~graph isomorphism is possible with constant number of queries, and (b)~any property is testable with constant number of queries in the respective families of graphs. The latter means that there is an algorithm that performs a constant number of queries and succeeds with constant probability. Here, `constant' means independent of $n$.

For graphs which are forests and outerplanar (which include forests), Kusumoto and Yoshida~\cite{KY14} and Babu, Khoury, and Newman~\cite{BKN16} argued, respectively,  the following in the centralized property testing model: (a)~graph isomorphism is possible with $\poly\log n$ number of queries, and (b)~any property is testable with the same number of queries in the respective families of graphs. The latter means that there is an algorithm that performs $\poly\log n$ number of queries and succeeds with constant probability.

Their results are summarized in the following theorems.

\begin{theorem}[{\cite[Thm.~3.2,~3.3 ]{NS13}}]\label{thm:newmans}
 Given an oracle access to a minor-free graph with maximum degree $d=O(1)$, any graph property is testable with constant number of  queries. This testing algorithm succeeds with constant probability. Specifically, testing isomorphism of two such graphs can be done in a constant number of queries and with constant probability.
\end{theorem}

\begin{theorem}[{\cite[Thm.~1.3,~1.1 ]{KY14}, \cite[Thm.~4.3,~4.2]{BKN16}}]\label{thm:forests}
 Given an oracle access to a $k$-edge-outerplanar graph, any graph property is testable with $\poly\log n$  queries.~\footnote{Definition~2.1 from~\cite{BKN16}: A graph $G$ is $1$-edge-outerplanar if it has a planar embedding in which all vertices of $G$ are on the outer face. A graph $G$ is $k$-edge-outerplanar if $G$ has a planar embedding such that if all edges on the exterior face are deleted, the connected components of the remaining graph are all $(k-1)$-edge-outerplanar.} This testing algorithm succeeds with constant probability. Specifically, testing isomorphism of two $k$-edge-outerplanar graph can be done in $\poly\log n$  queries and with constant probability.
\end{theorem}

Given the simulation argument above and Theorems~\ref{thm:newmans}, ~\ref{thm:forests}, we obtain the following corollaries.
\begin{corollary}\label{coro:hyper}
Any property is testable in \congest\ for minor-free graphs with maximum degree $d=O(1)$ within $O(D)$ rounds in the bounded-degree model.
Specifically, this holds for the graph isomorphism problem. \mnote{fill in. Argue that it is tight. Refer to this in the intro.}
\end{corollary}

\begin{corollary}\label{coro:outer}
Any property is testable in \congest\ for trees and $k$-edge-outerplanar graphs within $\tilde{O}(D)$ rounds in the general model.
Specifically, this holds for the graph isomorphism problem. \mnote{fill in. Argue that it is tight. Refer to this in the intro.}
\end{corollary} 

\fi


\section{Conclusions}
In this paper we provided both upper and lower bounds for the problem of testing isomorphism to a known graph in the \congest\ model.
The main question that we leave open is whether it is possible to improve the complexity of the randomized algorithm for the (exact) decision variant.
For the property testing variant, we provided, up to  poly-logarithmic factors, tight bounds for graphs which are dense, and for special families of sparse graphs.
For the intermediate spectrum of density,  the complexity of the property testing variant remains an open question.

As a further research we propose studying the (more complex) problem of testing isomorphism to a graph which is not given explicitly as a parameter.
One option is to provide only query access to this graph. Another option is to assume that the network is composed of two graphs and that each edge is marked according to which graph it belongs to.
Namely, for a set of nodes $V$ and two sets of edges $E_1$ and $E_2$, the goal is to test if $G_1 = (V, E_1)$ and $G_2 = (V, E_2)$ are isomorphic when we run on the network $G = (V, E_1\cup E_2)$.
\ifnum\icalp=0
\bibliographystyle{plain}
\fi
\bibliography{refs}

\appendix
\section{Probabilistic Preliminaries}\label{sec-prob}
\begin{theorem}[Multiplicative Chernoff's Bound] \label{thm-cher}
Let $X_1, \ldots, X_n$ be identical independent random variables ranging in $[0,1]$ and let $p = \mathbb{E}[X_1]$. Then, for every $\gamma \in (0,2]$, it holds that
\begin{equation}
\Pr\left[\left|\frac{1}{n} \cdot \sum_{i\in [n]} X_i - p\right| > \gamma \cdot p\right] < 2 \cdot e^{-\gamma^2 pn/4}\;.\label{eq-cher}
\end{equation}
\end{theorem}
\ifnum\icalp=1

\ifnum\icalp=1
\section{Figures}\label{app:figs}
\fi
\begin{figure}[h!]
    \centering
    \begin{subfigure}[b]{0.4\textwidth}
        \centering
        \includegraphics[width=\textwidth]{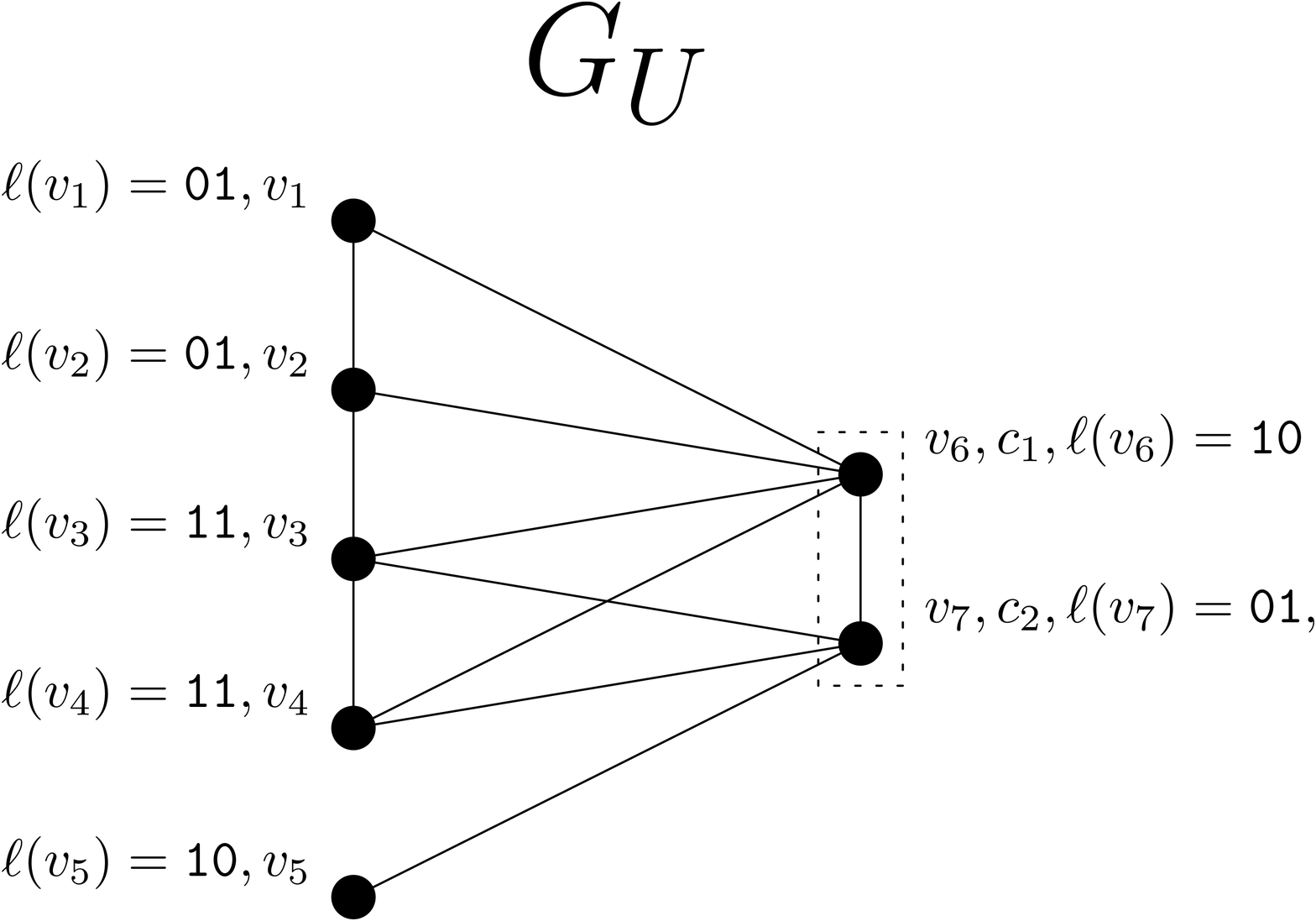}
        \caption{}
        \label{fig:Gu}
    \end{subfigure}
    \quad
    \begin{subfigure}[b]{0.4\textwidth}
        \centering
        \includegraphics[width=\textwidth]{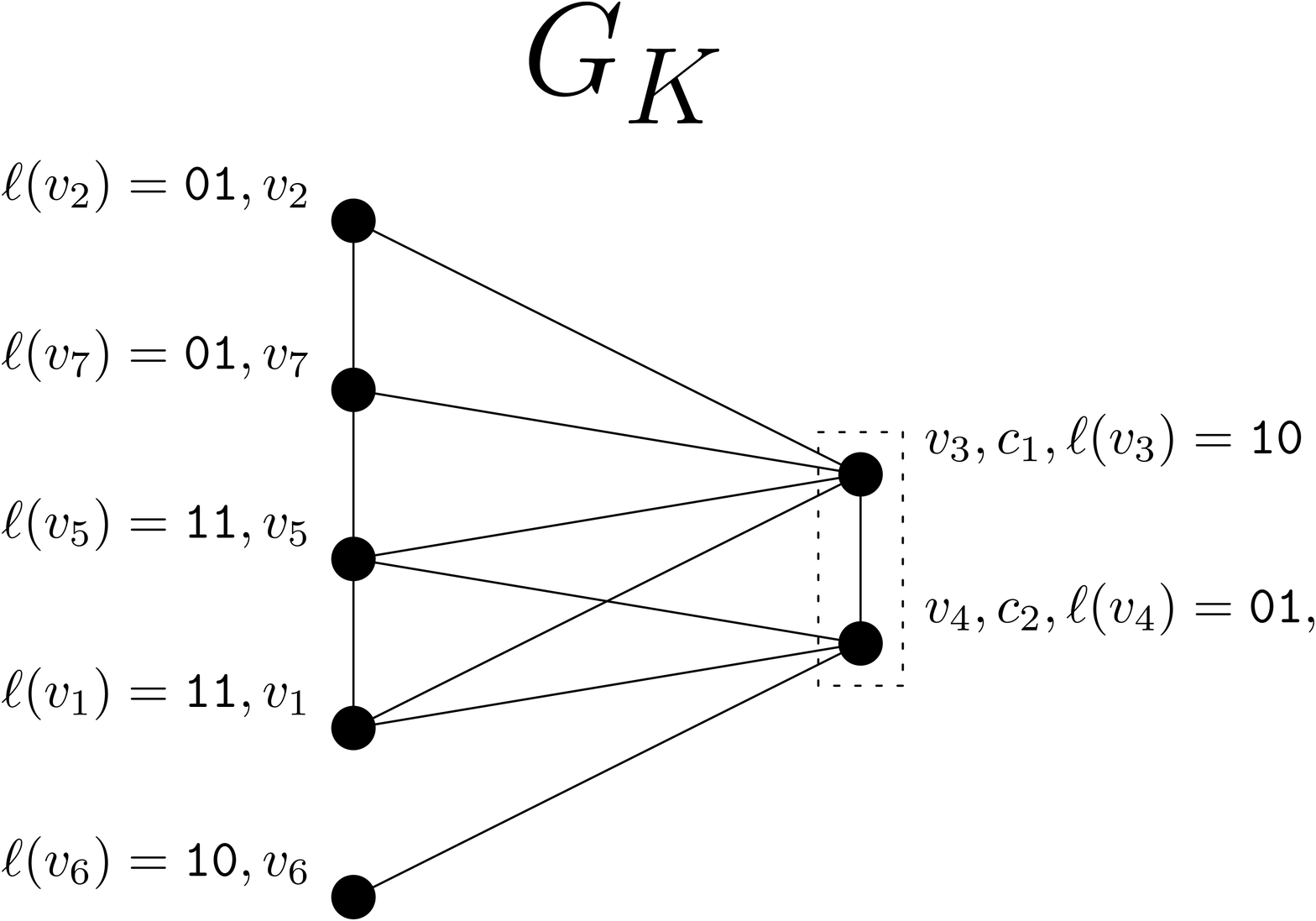}
        \caption{}
        \label{fig:Gk}
    \end{subfigure}
    \\
    \begin{subfigure}[b]{0.3\textwidth}
        \centering
        \includegraphics[width=\textwidth]{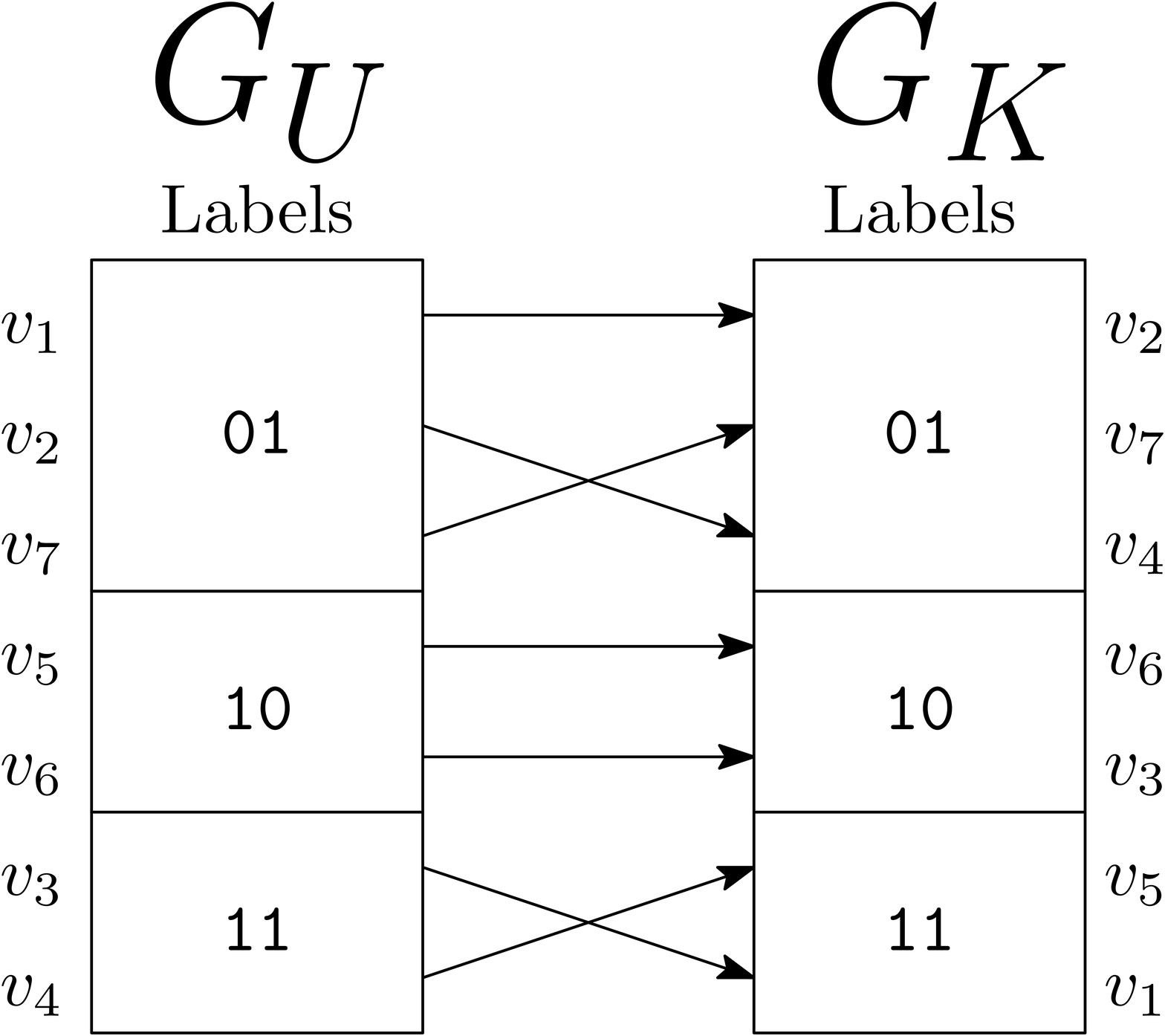}
        \caption{}
        \label{fig:mapping}
    \end{subfigure}
    \caption{
The unknown $G_U$ and known $G_K$ graphs are depicted in Sub-figure~\ref{fig:Gu}, and~\ref{fig:Gk}, respectively. The $C$ sequence consists of the nodes $c_1=v_6$ and $c_2=v_7$. Hence, the label of a node $v$ is a binary string $\ell_C^{G_U}=(b_2,b_1)\in\{0,1\}^2$, where $b_i=1 \Leftrightarrow c_i \in N(v)$, and the same for $G_K$. In the figures, we removed the super (sub) scripts as they are clear from the context. The induced labels are depicted next to the node number. Note that all the nodes have a neighbor in $C$, hence the label class ${\tt 00}$ is empty. In Sub-figure~\ref{fig:mapping} two tables are depicted: (1)~label assignment to the nodes of $G_U$, and (2)~label assignment to the nodes of $G_K$. The order of which the nodes are presented is according to the (single) isomorphism function between the two graphs. One can observe that the number of nodes per each corresponding label-classes, e.g., a buckets, is the same, as well as buckets which have the same ${\rm msb}$ bit, e.g., a cluster of labels. A random bijection that preserves the bucket of labels is depicted by the arrows. We show in Thm.~\ref{thm:findiso} that such a bijection yields graphs which are close to being isomorphic.}
\label{fig:labels}
\end{figure} 
\fi

\section{Detailed Proof of Thm.~\ref{thm:lb}}\label{app:liftassumplb}
In Theorem~\ref{thm:lb} we showed a family of graphs  for which we proved that any (randomized) tester errs with probability at least $1/3$. 
We assumed that the output of each node is invariant under re-labeling of the nodes and port numbers. In this section we re-prove this theorem without making this assumption. 
Given $D$ we construct $\cal{H}_L$ similarly to the construction of $\cal{H}$ (as in the less detailed proof of Theorem~\ref{thm:lb}) only that we specify the labels of the nodes and the port numbers in each graph.  

We begin with specifying the port numbers of the graphs. 
We assume without loss of generality that $D$ is even.
Port numbers of the path $p$ are set to be so that the edges adjacent to the graphs on each endpoint of $p$ are connected via port number $1$, and interior nodes are connected in an alternating fashion, that is, port $2$ is connected to port $1$, etc. This port assignment is possible since the length of $p$ is even.
Port numbering assignment for $G_1$ and $G_2$ may be arbitrary but fixed. The port in which the path $p$ is connected to either graphs is fixed as well. 

We next describe the labeling of the nodes in each of the graphs. 
The are three, mutually disjoint, sets of labels $A,B$ and $L$ where $|A| = |B| = n$ and $|L| = D-2$.
The labels from $A$ and $B$ are assigned to $G_1$ and $G_2$ in an arbitrary but fixed manner. 
This yields $4$ labeled graphs $G_i(S)$ for $i\in\{1,2\}$ and $S \in \{A,B\}$. 
For the path $p$ we consider \emph{two} label assignments: one where the labels in $L$ are assigned to interior nodes in $p$ in ascending order and one in descending order. 
We denote these two labeled paths by $p_a$ and $p_b$, respectively.
In turn, $\cal{H}_L$ has $2$ labeled graph for each graph in $\cal{H}$ (taking into account the labels assignment of $p$). 
These labeled graphs are denoted by $G_{i,j}(S_1,S_2,o)$ for $i,j \in \{1,2\}$, and $S_1,S_2 \in \{A,B\}$ s.t. $S_1 \neq S_2$, and $o \in \{a,d\}$ which denotes the orientation of the label assignment of $p$.
Note that the same label does not appear more than once in any of the graphs. 

We now define the set of inputs to the distributed tester: (1)~we fix the known graph $G_K$ to be $G_{1,2}(A,B,a)$, (2)~the unknown graph $G_U$ is chosen randomly from the distribution $\tilde{\cal{H}}$ which is defined as follows: the probability to obtain each of the four graphs in $G'_{1,2} \triangleq\{G_{1,2}(A,B,o), G_{1,2}(B,A,o) \mid o \in \{a,d\}\}$ is $1/6$ and each of the graphs in $G'_{1,1} \triangleq\{G_{1,1}(A,B,o) \mid o \in \{a,d\}\}$, and $G'_{2,2} \triangleq\{G_{2,2}(A,B,o) \mid o \in \{a,d\}\}$ is obtained with probability $1/12$.

We consider three cases. 
The first case is when $\mathcal{A}$ outputs (correctly) \accept\ whenever $G_U \in G'_{1,2}$.
By the definition of the problem this implies that all nodes output \accept.
Since the number of rounds of the tester is at most $D/3$, it follows that all the nodes in the graphs in $G'_{1,1}$, and $G'_{2,2}$, output \accept.
To see this observe consider any node $v$ in any one of the graphs in $G'_{1,1}$ and $G'_{2,2}$.
It is not hard to see that there exists a graph in $G'_{1,2}$ in which there exists a node, $u$, such that the $D/3$-hop neighborhoods of $v$ and $u$ are exactly the same (also when taking into account the labels on the nodes and the port numbers). 
Thus, in this case, $\mathcal{A}$ errs on all the graphs in $G'_{1,1}$ and $G'_{2,2}$, i.e., an error of $1/3$.

The second case is when $\mathcal{A}$ outputs \accept\ on exactly three of the graphs in $G'_{1,2}$ (that is, when $G_U$ equals to one of these three graphs).
Assume that the three graphs are $G_{1,2}(A,B,a)$, $G_{1,2}(A,B,b)$ and $G_{1,2}(B,A,b)$.
From the same reasoning as above, this implies that $\mathcal{A}$ outputs \accept\ also on $G_{1,1}(A,B,a)$ and $G_{2,2}(B,A,a)$. 
Thus, in this case the error probability is at least $1/3$ as well. 
The same analysis goes through for any three graphs in $G'_{1,2}$.

The third case is when $\mathcal{A}$ outputs \accept\ on at most two of the graphs in $G'_{1,2}$.
In this case it is immediate to see that the error probability is at least $1/3$. 

By Yao's principle it is implied that any randomized tester that performs at most $D/3$ rounds must err with probability at least $1/3$ as well.
By a straightforward amplification argument, this implies that any randomized tester that performs at most $D/9$ rounds must err with probability strictly greater than $1/3$~\footnote{To see this, assume towards contradiction that there exists a tester that succeeds with probability $2/3$ and performs at most $D/9$ rounds. Then by running this tester three times and taking a majority vote we obtain a tester that performs at most $D/3$ rounds and errs with probability less than $1/3$ in contradiction to what we showed.}, as desired.

\section{Missing proofs}\label{sec:mpfs}
\subsection{Proof of Claim~\ref{clm-sep}}
\cseven
\subsection{Proof of Lemma~\ref{lem-close}}
\leight

\subsection{Proof of Claim~\ref{lemma8}}
\cnine

\ifnum\icalp=1
\subsection{Proof of Claim~\ref{clm:1}}
\cthirt 
\fi

\end{document}